\documentclass[a4paper,10pt]{article}
\newcommand{\colorprint}[0]{}
\usepackage[utf8]{inputenc}
\usepackage{pstricks,pst-tree,pst-node} 
\usepackage{subcaption}
\usepackage{amsmath}
\usepackage{amsfonts}
\usepackage{amssymb}
\usepackage{amsthm}
\usepackage{mathrsfs}
\usepackage{amsthm}
\usepackage{subcaption}
\usepackage[colorlinks,urlcolor=magenta,linkcolor=blue,citecolor=teal]{hyperref}
\theoremstyle{plain}
\newtheorem{theorem}{Theorem}
\newtheorem{prop}[theorem]{Proposition}
\newtheorem{lemma}[theorem]{Lemma}
\newtheorem{corollary}[theorem]{Corollary}
\renewcommand*{\bar}[1]{\overline{#1}}
\usepackage{boxedminipage,verbatim,float,caption}
\newcommand{\pbDef}[3]{%
	\noindent
	\begin{center}
		\begin{boxedminipage}{ \columnwidth}
			\textbf{#1}\\[5pt]
			\textbf{Input:}  #2\\
			\textbf{Output:}  #3
		\end{boxedminipage}
	\end{center}
}

\usepackage{authblk}

\title{Proper-walk connection number of graphs\thanks{Research supported by the Danish research council under grant number DFF-7014-00037B.}}
\author[1]{Jørgen Bang-Jensen}
\author[2]{Thomas Bellitto\thanks{This work was conducted during a postdoc of this author at the University of Southern Denmark. This author is now supported by the European Research Council (ERC) under the European Union’s Horizon 2020 research and innovation programme Grant Agreement 714704.}}
\author[1,3]{Anders Yeo}
\affil[1]{\footnotesize{Department of Mathematics and Computer Science, University of Southern Denmark, Odense, Denmark.}}
\affil[2]{\footnotesize{Faculty of Mathematics, Informatics and Mechanics, University of Warsaw, Poland.}}
\affil[3]{\footnotesize{Department of Pure and Applied Mathematics, University of Johannesburg, Auckland Park, 2006 South Africa.}}
\date{}
\begin{document}

\maketitle

\begin{abstract}
This paper studies the problem of proper-walk connection number: given an undirected connected graph, our aim is to colour its edges with as few colours as possible so that there exists a properly coloured walk between every pair of vertices of the graph \textit{i.e.} a walk that does not use consecutively two edges of the same colour. The problem was already solved on several classes of graphs but still open in the general case. We establish that the problem can always be solved in polynomial time in the size of the graph and we provide a characterization of the graphs that can be properly connected with $k$ colours for every possible value of $k$.

\end{abstract}

\section{Introduction}

Let $G=(V,E)$ be an undirected graph. An edge-colouring of $G$ is a function $c:E\mapsto [1,n]$. Several kinds of edge-colourings have been studied but the ones that receive the most attention in the literature are undoubtedly proper edge-colourings: we say that an edge-colouring is proper if and only if two adjacent edges never receive the same colour. The number of colours required for a proper colouring of the edges of a graph is called the chromatic index of the graph and has been studied in many contexts and on many classes of graphs.
In 1976, Chen and Daykin have introduced in \cite{Daykin} the notion of properly coloured walks: a walk in an edge-coloured graph $G$ is said to be properly coloured if and only if it does not use consecutively two edges of the same colour. If $G$ itself is properly edge-coloured, every walk in $G$ is a properly coloured walk but the definition becomes non-trivial for improperly coloured graphs. To illustrate this, we recall that edge-coloured graphs can be seen as a powerful generalization of directed graphs \cite{bookjbjgutin}: indeed, if $D$ is a directed graph, we can subdivide every arc $uv$ of $G$ by inserting a vertex $w_{uv}$ and we can replace the arc $uv$ by a red edge $uw_{uv}$ and a blue edge $w_{uv}v$. We obtain a 2-edge-coloured bipartite undirected graph with a similar set of possible walks. We refer the reader to \cite{GutinKim} for a survey on properly coloured cycles and paths.

Edge-coloured graphs are an example of walk-restricted graphs. Indeed, there are many application fields of graph theory that require to investigate or solve optimization problems within sets of walks that are much more restricted than the set of all the possible walks in a graph. This leads to the definition and the study of several models that restrict the walks in a graph. Other famous examples of such graphs include forbidden-transition graphs, where only certain pairs of adjacent edges may be used consecutively and anti-directed walks in a directed graphs where the walks have to alternate between forward and backward arcs. 

Many concepts of graph theory can be extended to walk-restricted graphs. Walk-restricted graphs can provide insight on structural properties of the underlying unrestricted graph (see for example \cite{antistrong} where anti-directed walks are used to study 2-detachments), or can be used to model practical situations. For example, forbidden-transition graphs are used to solve routing problems in telecommunication networks \cite{ahmed} or in road networks \cite{sctm} and edge-coloured graphs are used in bio-informatics in \cite{Dorninger}. In \cite{Sudakov}, Sudakov discusses how to measure the robustness of certain graph properties such as Hamiltonicity or connectivity and shows that it can sometimes be done by determining how many restrictions  have to be put on the walks in a graph for the graph to lose the property. Also note that a bipartite graph admits a strongly connected orientation if and only if it admits a connecting 2-edge-colouring. Indeed, let $G=(V,E)$ be a bipartite graph and let $(V_1,V_2)$ be a bi-partition of its vertices. The possible walks in an orientation $\overrightarrow G$ of $G$ are the same as in the edge-coloured graph $G_c$ where the arcs of $\overrightarrow G$ going from a vertex of $V_1$ to a vertex of $V_2$ are replaced by red edges and the others by blue edges.
We refer the reader to \cite{bookjbjgutin} for a more extensive discussion of edge-coloured graphs as well as other generalizations of graphs and of their applications.

In recent years, several papers have studied the connectivity of walk-restricted graphs. In \cite{antistrong}, Bang-Jensen et al. study antistrong connectivity of digraphs. In \cite{mcts}, Bellitto and Bergougnoux look for the smallest number of transitions required to connect every pair of vertices of a graph with compatible walks and prove that the problem is NP-complete. However, the best-studied model in the literature is edge-coloured graphs. In this paper, we investigate minimal requirements for undirected edge-coloured graphs to be connected by walks. Given an undirected connected graph, the question we study is to determine how many colours are required to colour its edges in such a way that every two vertices are connected by a properly coloured walk. This condition is thus weaker than the properness of the edge-colouring and can often be achieved with much fewer colours.

The minimum number of colours required to colour all the edges of a graph in such a way that every pair of vertices can be connected by a properly-coloured path was introduced in 2012 in \cite{Borozan2012} and is called the \textbf{proper connection number} of the graph. Determining the proper connection number of a graph has since been studied in several contexts, both with directed and undirected graphs, and with different definitions of connectivity that either require that the vertices are connected by properly-coloured elementary paths or that allow the vertices to be connected by walks that repeat vertices. 
While the definitions may vary slightly from a paper to an other, the proper connection number is generally defined as the number of colours necessary to connect the vertices with paths. Thus, following Melville and Goddard \cite{MelvilleGoddard}, we will talk about \textbf{proper-walk connection number} in the case of walks.

In most cases, determining the proper or proper-walk connection number of a graph has proved to be a challenging problem. In the directed case, Ducoffe et al. proved that it is already NP-complete to determine if there exists a 2-edge-colouring such that every pair of vertices is connected by properly-coloured paths. In the undirected case, several papers have studied necessary or sufficient conditions for graphs that can be connected with 2 colours, both in the case of walks and paths \cite{Brause1} \cite{Brause2} \cite{Brause3} \cite{MelvilleGoddard}, but no characterization of those graphs had emerged yet. The main contribution of this paper is to provide a polynomial-time algorithm that determines the proper-walk connection number of an undirected graph and returns an optimal connecting edge-colouring (Theorem \ref{mainthm}).

More formally, in the rest of this paper, we define an edge-coloured undirected graph $G_c=(V,E,c)$ as \textbf{properly connected} if and only if for every two vertices $u$ and $v$ in $V$, there exists a properly coloured \textbf{walk} between $u$ and $v$. In this case, we say that $c$ is a \textbf{connecting edge-colouring} of $G$. For example, the edge-coloured graph depicted in Figure \ref{mainexample} is properly connected. For example, the vertices $v_0$ and $v_2$ are connected by the properly coloured walk $(v_0,v_3, v_4, v_5, v_{13}, v_{12}, v_8, v_4, v_3, v_2)$.
Note that the vertices of the graph only have to be connected by walks and we can thus repeat vertices or edges.  The vertices of the graph of Figure \ref{mainexample} cannot all be connected by properly coloured elementary paths but we still consider the graph to be properly connected.

  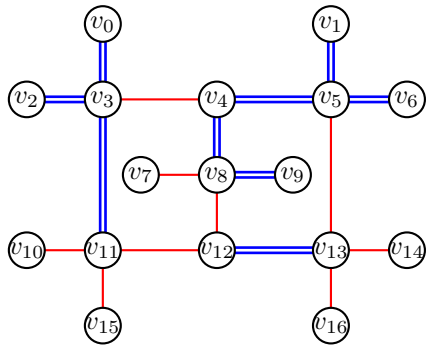
\begin{figure}[!h]
\[\begin{pspicture}(5.5,4.5)
\psline[linecolor=red](0.25,1.25)(1.25,1.25)
\colorprint{\psline[linecolor=blue, linewidth=3pt](0.25,3.25)(1.25,3.25)}
\psline[linecolor=red](1.25,0.25)(1.25,1.25)
\colorprint{\psline[linecolor=blue, linewidth=3pt](1.25,1.25)(1.25,3.25)}
\colorprint{\psline[linecolor=blue, linewidth=3pt](1.25,3.25)(1.25,4.25)}
\psline[linecolor=red](1.25,1.25)(2.75,1.25)
\psline[linecolor=red](1.25,3.25)(2.75,3.25)
\psline[linecolor=red](1.75,2.25)(2.75,2.25)
\psline[linecolor=red](2.75,1.25)(2.75,2.25)
\colorprint{\psline[linecolor=blue, linewidth=3pt](2.75,2.25)(2.75,3.25)}
\colorprint{\psline[linecolor=blue, linewidth=3pt](2.75,2.25)(3.75,2.25)}
\colorprint{\psline[linecolor=blue, linewidth=3pt](2.75,1.25)(4.25,1.25)}
\colorprint{\psline[linecolor=blue, linewidth=3pt](2.75,3.25)(4.25,3.25)}
\psline[linecolor=red](4.25,0.25)(4.25,1.25)
\psline[linecolor=red](5.25,1.25)(4.25,1.25)
\psline[linecolor=red](4.25,1.25)(4.25,3.25)
\colorprint{\psline[linecolor=blue, linewidth=3pt](4.25,3.25)(4.25,4.25)}
\colorprint{\psline[linecolor=blue, linewidth=3pt](4.25,3.25)(5.25,3.25)}

\psline[linecolor=blue, linewidth=1pt](0.25,3.25)(1.25,3.25)
\psline[linecolor=blue, linewidth=1pt](1.25,1.25)(1.25,3.25)
\psline[linecolor=blue, linewidth=1pt](1.25,3.25)(1.25,4.25)
\psline[linecolor=blue, linewidth=1pt](2.75,2.25)(2.75,3.25)
\psline[linecolor=blue, linewidth=1pt](2.75,2.25)(3.75,2.25)
\psline[linecolor=blue, linewidth=1pt](2.75,1.25)(4.25,1.25)
\psline[linecolor=blue, linewidth=1pt](2.75,3.25)(4.25,3.25)
\psline[linecolor=blue, linewidth=1pt](4.25,3.25)(4.25,4.25)
\psline[linecolor=blue, linewidth=1pt](4.25,3.25)(5.25,3.25)

\colorprint{\psline[linecolor=white, linewidth=1pt](0.25,3.25)(1.25,3.25)
\psline[linecolor=white, linewidth=1pt](1.25,1.25)(1.25,3.25)
\psline[linecolor=white, linewidth=1pt](1.25,3.25)(1.25,4.25)
\psline[linecolor=white, linewidth=1pt](2.75,2.25)(2.75,3.25)
\psline[linecolor=white, linewidth=1pt](2.75,2.25)(3.75,2.25)
\psline[linecolor=white, linewidth=1pt](2.75,1.25)(4.25,1.25)
\psline[linecolor=white, linewidth=1pt](2.75,3.25)(4.25,3.25)
\psline[linecolor=white, linewidth=1pt](4.25,3.25)(4.25,4.25)
\psline[linecolor=white, linewidth=1pt](4.25,3.25)(5.25,3.25)}

 \pscircle[fillstyle=solid, fillcolor=white](0.25,1.25){0.25}
 \pscircle[fillstyle=solid, fillcolor=white](0.25,3.25){0.25}
 \pscircle[fillstyle=solid, fillcolor=white](1.25,0.25){0.25}
 \pscircle[fillstyle=solid, fillcolor=white](1.25,1.25){0.25}
 \pscircle[fillstyle=solid, fillcolor=white](1.25,3.25){0.25}
 \pscircle[fillstyle=solid, fillcolor=white](1.25,4.25){0.25}
 \pscircle[fillstyle=solid, fillcolor=white](1.75,2.25){0.25}
 \pscircle[fillstyle=solid, fillcolor=white](2.75,1.25){0.25}
 \pscircle[fillstyle=solid, fillcolor=white](2.75,2.25){0.25}
 \pscircle[fillstyle=solid, fillcolor=white](2.75,3.25){0.25}
 \pscircle[fillstyle=solid, fillcolor=white](3.75,2.25){0.25}
 \pscircle[fillstyle=solid, fillcolor=white](4.25,0.25){0.25}
 \pscircle[fillstyle=solid, fillcolor=white](4.25,1.25){0.25}
 \pscircle[fillstyle=solid, fillcolor=white](4.25,3.25){0.25}
 \pscircle[fillstyle=solid, fillcolor=white](4.25,4.25){0.25}
 \pscircle[fillstyle=solid, fillcolor=white](5.25,1.25){0.25}
 \pscircle[fillstyle=solid, fillcolor=white](5.25,3.25){0.25}
 
\rput(0.25,1.25){$v_{10}$}
\rput(0.25,3.25){$v_2$}
\rput(1.25,0.25){$v_{15}$}
\rput(1.25,1.25){$v_{11}$}
\rput(1.25,3.25){$v_3$}
\rput(1.25,4.25){$v_0$}
\rput(1.75,2.25){$v_7$}
\rput(2.75,1.25){$v_{12}$}
\rput(2.75,2.25){$v_8$}
\rput(2.75,3.25){$v_4$}
\rput(3.75,2.25){$v_9$}
\rput(4.25,0.25){$v_{16}$}
\rput(4.25,1.25){$v_{13}$}
\rput(4.25,3.25){$v_5$}
\rput(4.25,4.25){$v_1$}
\rput(5.25,1.25){$v_{14}$}
\rput(5.25,3.25){$v_6$}

\end{pspicture}\]  
\caption{An example of properly connected edge-coloured graph. For readability in black and white, blue edges are represented with a double line.}
\label{mainexample}
\end{figure}

In this paper, the problem of connecting edge-colouring is defined as follows:

\pbDef{Connecting edge-colouring}
{A connected undirected graph $G=(V,E)$.}
{The smallest number of colours $k$ such that there exists a colouring function $c:E\mapsto [1,k]$ such that $G_c=(V,E,c)$ is properly connected.}

\subsection*{Terminology and definitions}

This paper follows the notation of \cite{bookjbjgutin}.

A \textbf{walk} of \textbf{length} $l$ is a sequence $v_0,v_1,\dots,v_l$ of adjacent vertices where $v_0$ and $v_l$ are the \textbf{end vertices}. If $v_0=v_l$, we say that the walk is \textbf{closed}.
A \textbf{path}, sometimes referred to as \textbf{elementary path} to avoid any ambiguity, is a walk whose vertices are all different and a \textbf{cycle} or \textbf{elementary cycle} is a closed walk whose vertices are all different except the end vertices.

The \textbf{distance} between two sets of vertices $S_1$ and $S_2$ is the length of a shortest walk that has an end vertex in $S_1$ and one in $S_2$. This definition allows for example to define the distance between two vertices or between a vertex and an edge or a path.

An \textbf{ear-decomposition} of a graph $G$ is an ordered set $(C,P_1,\dots,P_k)$ where 

$\bullet$ $C$ is a cycle in $G$ and each $P_i$ is a path or a cycle;

$\bullet$ $C$ and the $P_i$ partition the edges of the graph;

$\bullet$ for every $i$, the end vertices of $P_i$ are vertices of $C,P_1,\dots,P_{i-1}$ and its other vertices are not.
\\
It is well known that the 2-edge-connected graphs are exactly those that admit an ear-decomposition. Every cycle in a 2-edge-connected graph can be used as the starting cycle in an ear-decomposition. 

\section{The cases $k\neq 2$}

\subsection{Trivial bounds}

The number of colours required to connect a graph of $n$ vertices can be anywhere between 1 and $n-1$. Every connected graph can be connected with $n-1$ colours, for example by choosing a spanning tree and giving a different colour to all its edges. Complete graphs can be properly connected with only one colour while $n-1$ colours are required to connect a star. 

The number of colours required to connect a graph is also bounded by the chromatic index of the graph: indeed, if no two adjacent edges have the same colour, every walk in the graph is properly coloured but this condition is far from necessary. For example, the graphs with highest chromatic indexes are the complete graphs but they are those that require the fewest colours to be connected.

\subsection{Trees}

The number of colours required to connect a tree is exactly its maximum degree $\Delta$. Indeed, a greedy colouring of the edges of a rooted tree by order of increasing depth provides a proper edge-colouring of the tree using only $\Delta$ colours and thereby proves that $\Delta$ colours are enough to connect the tree. Conversely, if we colour the edges of the tree with fewer than $\Delta$ colours, every vertex $u$ of degree $\Delta$ will have two adjacent edges, say $uv$ and $uw$, with same colour and there are no properly coloured walks between $v$ and $w$.

\subsection{The case $k=1$}

Determining if a graph can be connected with only one colour comes down to determining if the graph is complete and can be done in $O(|V|+|E|)$. 

Indeed, if the graph is complete, all the vertices can be connected by a walk of length one and one colour is enough to connect the graph. Otherwise, two non-adjacent vertices $u$ and $v$ require a walk of length at least 2 and the first two edges of the walk must have different colours.

\subsection{Complexity of connecting $k$-edge-colouring with $k\geqslant 3$}

\begin{theorem} Any graph with a cycle can be connected with 3 colours.\label{lemmacycle}\end{theorem}

Note that this theorem already appears in \cite{MelvilleGoddard}. However, we leave it in this paper for the sake of completeness. Indeed, our proof is constructive and we would like our paper to provide a general algorithm for colouring optimally any undirected graph. 

\begin{proof}

Let $C$ be an elementary cycle in the graph and let $V_C$ and $E_C$ be the vertices and edges of $C$. Let $\chi$ be a proper edge-colouring of $C$ using at most 3 colours.

We set $G'=G\setminus E_C$. Since $G$ is connected, we know that every connected component $C_i$ of $G'$ contains a vertex $v_i$ of $V_C$. For every connected component of $G'$, we 
know that $v_i$ has two incident edges in $E_C$. Let $a_i\in\{1,2,3\}$ be the colour that is not used by $\chi$ to colour any incident edge of $v_i$ in $C$ and let $b_i\neq a_i$ be another colour of $\{1,2,3\}$. We extend $\chi$ by using colour $a_i$ on every edge of $C_i$ at even distance from $v_i$ in $G'$ and colour $b_i$ on every edge at odd distance.
This construction is illustrated in Figure \ref{figlemmacycle}. 

We claim that the resulting colouring $\chi$ connects the graph. Indeed, let $u$ and $w$ be vertices of the graph and let $C_i$ and $C_j$ be their respective components in $G'$. 
The vertices $u$ and $w$ can be connected by:
\begin{itemize}
 \item going from $u$ to $v_i$ using a shortest walk in $G'$
 \item using the cycle $c$ to connect $v_i$ and $v_j$ (if $v_i=v_j$, we use the entire cycle and not an empty walk);
 \item going from $v_j$ to $w$ using a shortest walk in $G'$.\qedhere
\end{itemize}

  \begin{figure}[!h]
\begin{subfigure}{0.48\linewidth}
\[\begin{pspicture}(-1,0)(5,3.16)
\psset{unit=0.9cm}
\psline[linecolor=gray, linestyle=dashed](-0.81,0.66)(0.19,1.84)
\psline[linecolor=gray, linestyle=dashed](-0.81,0.66)(0.19,0.66)
\psline[linecolor=gray, linestyle=dashed](0.19,0.66)(0.19,1.84)
\psline[linecolor=gray, linestyle=dashed](0.19,0.66)(1.19,0.66)
\psline[linecolor=gray, linestyle=dashed](0.19,1.84)(1.19,1.84)
\psline[linecolor=red](1.19,0.66)(1.19,1.84)
\psline[linecolor=blue](1.19,0.66)(2.31,0.3)
\colorprint{\psline[linecolor=blue, linewidth=3pt](1.19,0.66)(2.31,0.3)
\psline[linecolor=white](1.19,0.66)(2.31,0.3)}
\psline[linecolor=red](2.31,2.2)(3,1.25)
\psline[linecolor=blue](2.31,2.2)(1.19,1.84)
\colorprint{\psline[linecolor=blue, linewidth=3pt](2.31,2.2)(1.19,1.84)
\psline[linecolor=white](2.31,2.2)(1.19,1.84)}
\colorprint{\psline[linecolor=green, linewidth=5pt](2.31,0.3)(3,1.25)
\psline[linecolor=white, linewidth=3pt](2.31,0.3)(3,1.25)}
\psline[linecolor=green](2.31,0.3)(3,1.25)
\psline[linecolor=gray, linestyle=dashed](2.31,2.2)(1.4,2.91)
\psline[linecolor=gray, linestyle=dashed](2.31,2.2)(3.02,2.91)
\psline[linecolor=gray, linestyle=dashed](3.02,2.91)(4.02,2.91)
\psline[linecolor=gray, linestyle=dashed](3,1.25)(4,1.25)
\psline[linecolor=gray, linestyle=dashed](4,1.25)(4.71,1.96)
\psline[linecolor=gray, linestyle=dashed](4,1.25)(4.71,0.54)
\psline[linecolor=gray, linestyle=dashed](4.71,1.96)(4.71,0.54)

 \pscircle[fillstyle=solid, fillcolor=white](-0.81,0.66){0.25}
 \pscircle[fillstyle=solid, fillcolor=white](0.19,0.66){0.25}
 \pscircle[fillstyle=solid, fillcolor=white](0.19,1.84){0.25}
 \pscircle[fillstyle=solid, fillcolor=white](1.19,0.66){0.25}
 \pscircle[fillstyle=solid, fillcolor=white](1.19,1.84){0.25}
 \pscircle[fillstyle=solid, fillcolor=white](2.31,2.2){0.25}
 \pscircle[fillstyle=solid, fillcolor=white](1.4,2.91){0.25}
 \pscircle[fillstyle=solid, fillcolor=white](3.02,2.91){0.25}
 \pscircle[fillstyle=solid, fillcolor=white](4.02,2.91){0.25}
 \pscircle[fillstyle=solid, fillcolor=white](2.31,0.3){0.25}
 \pscircle[fillstyle=solid, fillcolor=white](3,1.25){0.25}
 \pscircle[fillstyle=solid, fillcolor=white](4,1.25){0.25}
 \pscircle[fillstyle=solid, fillcolor=white](4.71,1.96){0.25}
 \pscircle[fillstyle=solid, fillcolor=white](4.71,0.54){0.25}

\end{pspicture}\]
\caption{The first step is to choose a cycle \\ $C$ in the graph and colour it properly.}
\end{subfigure}
\begin{subfigure}{0.48\linewidth}
\[\begin{pspicture}(-1,0)(5,3.16)
\psset{unit=0.9cm}
\psline[linecolor=black](-0.81,0.66)(0.19,0.66)
\psline[linecolor=black](0.19,0.66)(0.19,1.84)
\psline[linecolor=black](0.19,0.66)(1.19,0.66)
\psline(-0.81,0.66)(0.19,1.84)
\psline[linecolor=black](0.19,1.84)(1.19,1.84)
\psline[linecolor=black](2.31,2.2)(1.4,2.91)
\psline[linecolor=black](2.31,2.2)(3.02,2.91)
\psline[linecolor=black](3.02,2.91)(4.02,2.91)
\psline[linecolor=black](3,1.25)(4,1.25)
\psline[linecolor=black](4,1.25)(4.71,1.96)
\psline(4,1.25)(4.71,0.54)
\psline[linecolor=black](4.71,1.96)(4.71,0.54)

 \pscircle[fillstyle=solid, fillcolor=white](-0.81,0.66){0.25}
 \pscircle[fillstyle=solid, fillcolor=white](0.19,0.66){0.25}
 \pscircle[fillstyle=solid, fillcolor=white](0.19,1.84){0.25}
 \pscircle[fillstyle=solid, fillcolor=white](1.19,0.66){0.25}
 \pscircle[fillstyle=solid, fillcolor=white](1.19,1.84){0.25}
 \pscircle[fillstyle=solid, fillcolor=white](2.31,2.2){0.25}
 \pscircle[fillstyle=solid, fillcolor=white](1.4,2.91){0.25}
 \pscircle[fillstyle=solid, fillcolor=white](3.02,2.91){0.25}
 \pscircle[fillstyle=solid, fillcolor=white](4.02,2.91){0.25}
 \pscircle[fillstyle=solid, fillcolor=white](2.31,0.3){0.25}
 \pscircle[fillstyle=solid, fillcolor=white](3,1.25){0.25}
 \pscircle[fillstyle=solid, fillcolor=white](4,1.25){0.25}
 \pscircle[fillstyle=solid, fillcolor=white](4.71,1.96){0.25}
 \pscircle[fillstyle=solid, fillcolor=white](4.71,0.54){0.25}
 
 \rput(1.19,0.66){$v_1$}
 \rput(2.31,2.2){$v_2$}
 \rput(3,1.25){$v_3$}
 
\end{pspicture}\]
\caption{We create $G'$ by removing the edges of the chosen cycle and we pick arbitrarily a vertex of $C$ in each connected component of $G'$.}
\end{subfigure}

\begin{subfigure}{0.48\linewidth}
\[\begin{pspicture}(-1,0)(5,3.16)
\psset{unit=0.9cm}
\psline[linecolor=red](-0.81,0.66)(0.19,0.66)
\psline[linecolor=red](0.19,0.66)(0.19,1.84)
\colorprint{\psline[linecolor=green, linewidth=5pt](0.19,0.66)(1.19,0.66)
\psline[linecolor=white, linewidth=3pt](0.19,0.66)(1.19,0.66)}
\psline[linecolor=green](0.19,0.66)(1.19,0.66)
\colorprint{\psline[linecolor=green, linewidth=5pt](-0.81,0.66)(0.19,1.84)
\psline[linecolor=white, linewidth=3pt](-0.81,0.66)(0.19,1.84)}
\psline[linecolor=green](-0.81,0.66)(0.19,1.84)
\colorprint{\psline[linecolor=green, linewidth=5pt](0.19,1.84)(1.19,1.84)
\psline[linecolor=white, linewidth=3pt](0.19,1.84)(1.19,1.84)}
\psline[linecolor=green](0.19,1.84)(1.19,1.84)
\psline[linecolor=red, linestyle=dotted](1.19,0.66)(1.19,1.84)
\psline[linecolor=blue, linestyle=dotted](1.19,0.66)(2.31,0.3)
\psline[linecolor=red, linestyle=dotted](2.31,2.2)(3,1.25)
\psline[linecolor=blue, linestyle=dotted](2.31,2.2)(1.19,1.84)
\psline[linecolor=green, linestyle=dotted](2.31,0.3)(3,1.25)
\colorprint{\psline[linecolor=green, linewidth=5pt](2.31,2.2)(1.4,2.91)
\psline[linecolor=white, linewidth=3pt](2.31,2.2)(1.4,2.91)}
\psline[linecolor=green](2.31,2.2)(1.4,2.91)
\colorprint{\psline[linecolor=green, linewidth=5pt](2.31,2.2)(3.02,2.91)
\psline[linecolor=white, linewidth=3pt](2.31,2.2)(3.02,2.91)}
\psline[linecolor=green](2.31,2.2)(3.02,2.91)
\psline[linecolor=blue](3.02,2.91)(4.02,2.91)
\psline[linecolor=blue](3,1.25)(4,1.25)
\colorprint{\psline[linecolor=blue, linewidth=3pt](3.02,2.91)(4.02,2.91)
\psline[linecolor=white](3.02,2.91)(4.02,2.91)
\psline[linecolor=blue, linewidth=3pt](3,1.25)(4,1.25)
\psline[linecolor=white](3,1.25)(4,1.25)}
\psline[linecolor=red](4,1.25)(4.71,1.96)
\psline[linecolor=red](4,1.25)(4.71,0.54)
\psline[linecolor=blue](4.71,1.96)(4.71,0.54)
\colorprint{\psline[linecolor=blue, linewidth=3pt](4.71,1.96)(4.71,0.54)
\psline[linecolor=white](4.71,1.96)(4.71,0.54)}

 \pscircle[fillstyle=solid, fillcolor=white](-0.81,0.66){0.25}
 \pscircle[fillstyle=solid, fillcolor=white](0.19,0.66){0.25}
 \pscircle[fillstyle=solid, fillcolor=white](0.19,1.84){0.25}
 \pscircle[fillstyle=solid, fillcolor=white](1.19,0.66){0.25}
 \pscircle[fillstyle=solid, fillcolor=white](1.19,1.84){0.25}
 \pscircle[fillstyle=solid, fillcolor=white](2.31,2.2){0.25}
 \pscircle[fillstyle=solid, fillcolor=white](1.4,2.91){0.25}
 \pscircle[fillstyle=solid, fillcolor=white](3.02,2.91){0.25}
 \pscircle[fillstyle=solid, fillcolor=white](4.02,2.91){0.25}
 \pscircle[fillstyle=solid, fillcolor=white](2.31,0.3){0.25}
 \pscircle[fillstyle=solid, fillcolor=white](3,1.25){0.25}
 \pscircle[fillstyle=solid, fillcolor=white](4,1.25){0.25}
 \pscircle[fillstyle=solid, fillcolor=white](4.71,1.96){0.25}
 \pscircle[fillstyle=solid, fillcolor=white](4.71,0.54){0.25}
 
 \rput(1.19,0.66){$v_1$}
 \rput(2.31,2.2){$v_2$}
 \rput(3,1.25){$v_3$}

\end{pspicture}\]
\caption{In each component $C_i$, we colour the edges at even distance form $v_i$ with the colour that is not used by the edges incident to the $v_i$ in $C$ and the other edges with another colour.}
\end{subfigure}
\begin{subfigure}{0.48\linewidth}
\[\begin{pspicture}(-1,0)(5,3.16)
\psset{unit=0.9cm}
\psline[linecolor=red](-0.81,0.66)(0.19,0.66)
\psline[linecolor=red](0.19,0.66)(0.19,1.84)
\colorprint{\psline[linecolor=green, linewidth=5pt](0.19,0.66)(1.19,0.66)
\psline[linecolor=white, linewidth=3pt](0.19,0.66)(1.19,0.66)}
\psline[linecolor=green](0.19,0.66)(1.19,0.66)
\colorprint{\psline[linecolor=green, linewidth=5pt](-0.81,0.66)(0.19,1.84)
\psline[linecolor=white, linewidth=3pt](-0.81,0.66)(0.19,1.84)}
\psline[linecolor=green](-0.81,0.66)(0.19,1.84)
\colorprint{\psline[linecolor=green, linewidth=5pt](0.19,1.84)(1.19,1.84)
\psline[linecolor=white, linewidth=3pt](0.19,1.84)(1.19,1.84)}
\psline[linecolor=green](0.19,1.84)(1.19,1.84)
\psline[linecolor=red](1.19,0.66)(1.19,1.84)
\psline[linecolor=blue](1.19,0.66)(2.31,0.3)
\colorprint{\psline[linecolor=blue, linewidth=3pt](1.19,0.66)(2.31,0.3)
\psline[linecolor=white](1.19,0.66)(2.31,0.3)}
\psline[linecolor=red](2.31,2.2)(3,1.25)
\psline[linecolor=blue](2.31,2.2)(1.19,1.84)
\colorprint{\psline[linecolor=blue, linewidth=3pt](2.31,2.2)(1.19,1.84)
\psline[linecolor=white](2.31,2.2)(1.19,1.84)}
\colorprint{\psline[linecolor=green, linewidth=5pt](2.31,0.3)(3,1.25)
\psline[linecolor=white, linewidth=3pt](2.31,0.3)(3,1.25)}
\psline[linecolor=green](2.31,0.3)(3,1.25)
\colorprint{\psline[linecolor=green, linewidth=5pt](2.31,2.2)(1.4,2.91)
\psline[linecolor=white, linewidth=3pt](2.31,2.2)(1.4,2.91)}
\psline[linecolor=green](2.31,2.2)(1.4,2.91)
\colorprint{\psline[linecolor=green, linewidth=5pt](2.31,2.2)(3.02,2.91)
\psline[linecolor=white, linewidth=3pt](2.31,2.2)(3.02,2.91)}
\psline[linecolor=green](2.31,2.2)(3.02,2.91)
\psline[linecolor=blue](3.02,2.91)(4.02,2.91)
\colorprint{\psline[linecolor=blue, linewidth=3pt](3.02,2.91)(4.02,2.91)
\psline[linecolor=white](3.02,2.91)(4.02,2.91)}
\psline[linecolor=blue](3,1.25)(4,1.25)
\colorprint{\psline[linecolor=blue, linewidth=3pt](3,1.25)(4,1.25)
\psline[linecolor=white](3,1.25)(4,1.25)}
\psline[linecolor=red](4,1.25)(4.71,1.96)
\psline[linecolor=red](4,1.25)(4.71,0.54)
\psline[linecolor=blue](4.71,1.96)(4.71,0.54)
\colorprint{\psline[linecolor=blue, linewidth=3pt](4.71,1.96)(4.71,0.54)
\psline[linecolor=white](4.71,1.96)(4.71,0.54)}

 \pscircle[fillstyle=solid, fillcolor=white](-0.81,0.66){0.25}
 \pscircle[fillstyle=solid, fillcolor=white](0.19,0.66){0.25}
 \pscircle[fillstyle=solid, fillcolor=white](0.19,1.84){0.25}
 \pscircle[fillstyle=solid, fillcolor=white](1.19,0.66){0.25}
 \pscircle[fillstyle=solid, fillcolor=white](1.19,1.84){0.25}
 \pscircle[fillstyle=solid, fillcolor=white](2.31,2.2){0.25}
 \pscircle[fillstyle=solid, fillcolor=white](1.4,2.91){0.25}
 \pscircle[fillstyle=solid, fillcolor=white](3.02,2.91){0.25}
 \pscircle[fillstyle=solid, fillcolor=white](4.02,2.91){0.25}
 \pscircle[fillstyle=solid, fillcolor=white](2.31,0.3){0.25}
 \pscircle[fillstyle=solid, fillcolor=white](3,1.25){0.25}
 \pscircle[fillstyle=solid, fillcolor=white](4,1.25){0.25}
 \pscircle[fillstyle=solid, fillcolor=white](4.71,1.96){0.25}
 \pscircle[fillstyle=solid, fillcolor=white](4.71,0.54){0.25}

\end{pspicture}\]
\caption{The resulting colouring connects the graph.}
\end{subfigure}
\caption{An example of how to construct a connecting 3-edge-colouring of a graph with a cycle.}
\label{figlemmacycle}
\end{figure}
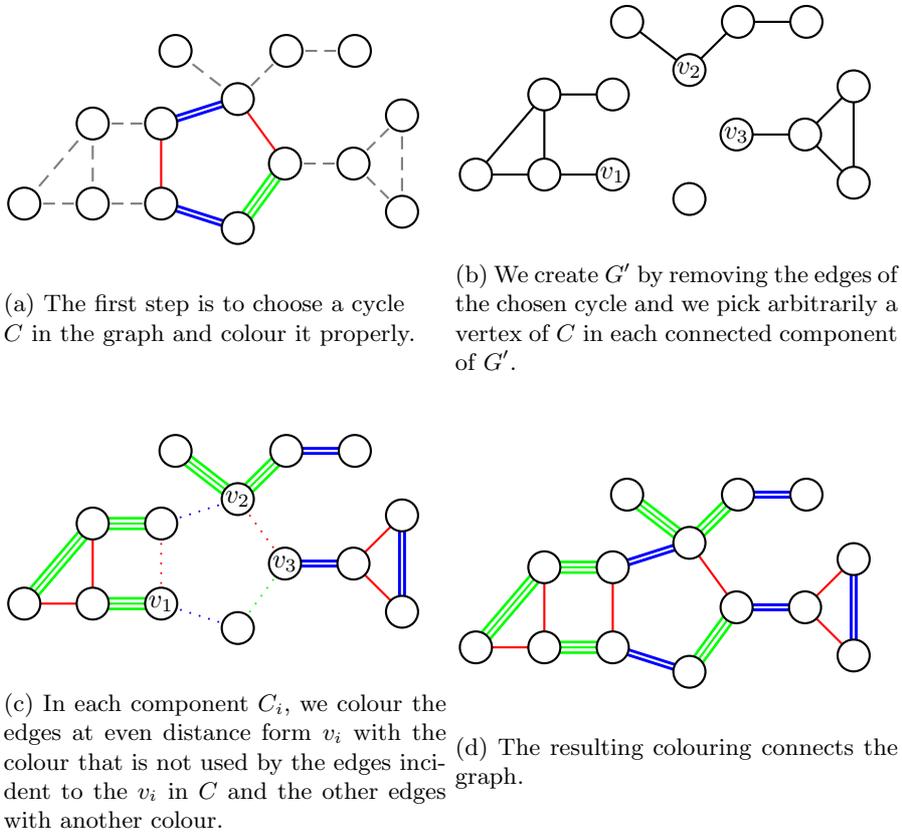
 \end{proof}

 The complexity of $k$-colouring for $k\geqslant 3$ quickly follows.
 
 \begin{corollary}
If $k\geqslant 3$, we can decide in polynomial time if a graph $G$ can be connected with $k$ colours. 
\end{corollary}

\begin{proof}
If the graph is a tree, the question comes down to deciding if the graph has a vertex of degree strictly greater than $k$, which is easy. Otherwise, the answer is always yes.
\end{proof}

 Hence, the only remaining case is $k=2$.
 
\section{Connecting 2-edge-colouring}

All the colourings we consider in this section are 2-colourings.

\subsection{Bipartite graphs}

We present in this subsection a characterization of the bipartite graphs that can be connected with two colours. This question has also been answered independently by \cite{MelvilleGoddard} and \cite{Ducoffe} but we keep it for the sake of completeness. Another reason for keeping the proof in the paper is that it illustrates a very nice correspondence between strong connectivity for bipartite digraphs and connecting 2-edge-colourings of bipartite graphs. This has lead to a number of nice proofs of results on 2-edge-coloured graphs, see \cite[Section 16.7]{bookjbjgutin}.

\begin{theorem}\label{thmbipartite}
 A bipartite graph $G$ can be connected with two colours if and only if it can be made 2-edge-connected by adding at most one edge.
\end{theorem}

\begin{proof} We prove the two implications separately:
\begin{itemize}
 \item [$\Rightarrow$:] Let $G$ be a graph and let us consider the tree $T(G)$ whose vertices are the 2-edge-connected components of $G$ and whose edges are the bridges of $G$. If $T(G)$ has only one leaf (and thus, only one vertex), this means that $G$ is 2-edge-connected. It $T(G)$ has two leaves, it is a path and one can 2-edge-connect it by adding an edge between the two leaves. Hence, if $G$ is not 2-edge-connected and cannot be made 2-edge-connected by adding only one edge, we know that $T(G)$ has at least three leaves. Let $e_1=u_1v_1$, $e_2=u_2v_2$ and $e_3=u_3v_3$ be bridges connecting three distinct leaves of $T(G)$ to the rest of the tree such that the vertices $v_i$ belong to the leaves.

Since $G$ is bipartite, if there is a walk of odd length between $u_1$ and $u_2$, we know that all the walks between $u_1$ and $u_2$ are of odd. Thus, every walk between $v_1$ and $v_2$ consists of $e_1$, a subwalk of odd length and $e_2$, which means that $e_1$ and $e_2$ must have the same colour in any connecting 2-edge-colouring. Similarly, if there is a walk of even length between $u_1$ and $u_2$, the edges $e_1$ and $e_2$ must have different colours. By applying this observation to $e_1$ and $e_3$ too, we find that if the distance between $u_1$ and $u_2$ has the same parity as the distance between $u_1$ and $u_3$, $e_2$ and $e_3$ must have the same colour and conversely, if these distances have different parity, $e_2$ and $e_3$ must have different colours. However, if the distance between $u_1$ and $u_2$ has the same parity as the distance between $u_1$ and $u_3$, this means that there exists a walk of even length between $u_2$ and $u_3$ (going through $u_1$), which means that $e_2$ and $e_3$ must have different colours for $v_2$ to be connected to $v_3$, contradicting the above. The same contradiction arises if the distance between $u_1$ and $u_2$ has different parity than the distance between $u_1$ and $u_3$. This is illustrated in Figure \ref{evencycles}. Hence, if $G$ cannot be made 2-edge-connected by adding an edge, then $G$ cannot be connected with 2 colours.
 
 \begin{figure}[!h]
\begin{subfigure}{\linewidth}
\begin{pspicture}(2.5,3.5)
\psset{unit=0.9cm}
\psline[linecolor=red](0.25,0.25)(1.25,1.25)
\psline[linecolor=blue](2.25,0.25)(1.25,1.25)
\colorprint{\psline[linecolor=blue, linewidth=3pt](2.25,0.25)(1.25,1.25)
\psline[linecolor=white](2.25,0.25)(1.25,1.25)}
\psline[linecolor=red](1.25,1.25)(0.25,2.25)
\psline[linecolor=blue](1.25,1.25)(2.25,2.25)
\colorprint{\psline[linecolor=blue, linewidth=3pt](1.25,1.25)(2.25,2.25)
\psline[linecolor=white](1.25,1.25)(2.25,2.25)}
\psline[linecolor=blue](0.25,2.25)(1.25,3.25)
\colorprint{\psline[linecolor=blue, linewidth=3pt](0.25,2.25)(1.25,3.25)
\psline[linecolor=white](0.25,2.25)(1.25,3.25)}
\psline[linecolor=red](2.25,2.25)(1.25,3.25)

 \pscircle[fillstyle=solid, fillcolor=white](0.25,0.25){0.25}
 \pscircle[fillstyle=solid, fillcolor=white](2.25,0.25){0.25}
 \pscircle[fillstyle=solid, fillcolor=white](1.25,1.25){0.25}
 \pscircle[fillstyle=solid, fillcolor=white](0.25,2.25){0.25}
 \pscircle[fillstyle=solid, fillcolor=white](2.25,2.25){0.25}
 \pscircle[fillstyle=solid, fillcolor=white](1.25,3.25){0.25}

\end{pspicture}
\begin{pspicture}(-1.4,0)(3.9,3.5)
\psset{unit=0.9cm}
\psline[linecolor=red](1.25,1.25)(0.25,2.25)
\psline[linecolor=blue](1.25,1.25)(2.25,2.25)
\colorprint{\psline[linecolor=blue, linewidth=3pt](1.25,1.25)(2.25,2.25)
\psline[linecolor=white](1.25,1.25)(2.25,2.25)}
\psline[linecolor=blue](0.25,2.25)(1.25,3.25)
\colorprint{\psline[linecolor=blue, linewidth=3pt](0.25,2.25)(1.25,3.25)
\psline[linecolor=white](0.25,2.25)(1.25,3.25)}
\psline[linecolor=red](2.25,2.25)(1.25,3.25)
\psline[linecolor=blue](0.25,2.25)(-1.16,2.25)
\colorprint{\psline[linecolor=blue, linewidth=3pt](0.25,2.25)(-1.16,2.25)
\psline[linecolor=white](0.25,2.25)(-1.16,2.25)}
\psline[linecolor=red](2.25,2.25)(3.61,2.25)

 \pscircle[fillstyle=solid, fillcolor=white](1.25,1.25){0.25}
 \pscircle[fillstyle=solid, fillcolor=white](-1.16,2.25){0.25}
 \pscircle[fillstyle=solid, fillcolor=white](3.61,2.25){0.25}
  \pscircle[fillstyle=solid, fillcolor=white](0.25,2.25){0.25}
 \pscircle[fillstyle=solid, fillcolor=white](2.25,2.25){0.25}
 \pscircle[fillstyle=solid, fillcolor=white](1.25,3.25){0.25}

\end{pspicture}
\begin{pspicture}(3.75,3.5)
\psset{unit=0.9cm}
\psline[linecolor=red](1.25,1.25)(0.25,2.25)
\psline[linecolor=blue](1.25,1.25)(2.25,2.25)
\colorprint{\psline[linecolor=blue, linewidth=3pt](1.25,1.25)(2.25,2.25)
\psline[linecolor=white](1.25,1.25)(2.25,2.25)}
\psline[linecolor=blue](0.25,2.25)(1.25,3.25)
\colorprint{\psline[linecolor=blue, linewidth=3pt](0.25,2.25)(1.25,3.25)
\psline[linecolor=white](0.25,2.25)(1.25,3.25)}
\psline[linecolor=red](1.25,1.25)(2.25,0.25)
\psline[linecolor=red](2.25,2.25)(1.25,3.25)
\psline[linecolor=red](2.25,2.25)(3.25,3.25)

 \pscircle[fillstyle=solid, fillcolor=white](1.25,1.25){0.25}
 \pscircle[fillstyle=solid, fillcolor=white](2.25,0.25){0.25}
 \pscircle[fillstyle=solid, fillcolor=white](3.25,3.25){0.25}
  \pscircle[fillstyle=solid, fillcolor=white](0.25,2.25){0.25}
 \pscircle[fillstyle=solid, fillcolor=white](2.25,2.25){0.25}
 \pscircle[fillstyle=solid, fillcolor=white](1.25,3.25){0.25}
\end{pspicture}
\caption{Three graphs that have three 2-edge-connected components: the induced $C_4$ and each of the two vertices of degree one. The bridges are therefore the edges connecting the vertices of degree one to the rest of the graph. Two bridges at odd distance must have the same colour and two bridges at even distance must have different colours.}
\label{evencycles1}
\end{subfigure}

\begin{subfigure}{\linewidth}
\begin{pspicture}(2.5,4)
\psset{unit=0.9cm}
\psline[linecolor=black](0.25,0.75)(1.25,1.75)
\psline[linecolor=black](2.25,0.75)(1.25,1.75)
\psline[linecolor=black](1.25,0.34)(1.25,1.75)
\psline[linecolor=black](1.25,1.75)(0.25,2.75)
\psline[linecolor=black](1.25,1.75)(2.25,2.75)
\psline[linecolor=black](0.25,2.75)(1.25,3.75)
\psline[linecolor=black](2.25,2.75)(1.25,3.75)

 \pscircle[fillstyle=solid, fillcolor=white](0.25,0.75){0.25}
 \pscircle[fillstyle=solid, fillcolor=white](2.25,0.75){0.25}
 \pscircle[fillstyle=solid, fillcolor=white](1.25,1.75){0.25}
 \pscircle[fillstyle=solid, fillcolor=white](1.25,0.34){0.25}
 \pscircle[fillstyle=solid, fillcolor=white](0.25,2.75){0.25}
 \pscircle[fillstyle=solid, fillcolor=white](2.25,2.75){0.25}
 \pscircle[fillstyle=solid, fillcolor=white](1.25,3.75){0.25}

\end{pspicture}
\begin{pspicture}(-1.4,0)(3.9,4)
\psset{unit=0.9cm}
\psline[linecolor=black](1.25,0.16)(1.25,1.75)
\psline[linecolor=black](1.25,1.75)(0.25,2.75)
\psline[linecolor=black](1.25,1.75)(2.25,2.75)
\psline[linecolor=black](0.25,2.75)(1.25,3.75)
\psline[linecolor=black](2.25,2.75)(1.25,3.75)
\psline[linecolor=black](0.25,2.75)(-1.16,2.75)
\psline[linecolor=black](2.25,2.75)(3.61,2.75)

 \pscircle[fillstyle=solid, fillcolor=white](1.25,0.34){0.25}
 \pscircle[fillstyle=solid, fillcolor=white](1.25,1.75){0.25}
 \pscircle[fillstyle=solid, fillcolor=white](-1.16,2.75){0.25}
 \pscircle[fillstyle=solid, fillcolor=white](3.61,2.75){0.25}
  \pscircle[fillstyle=solid, fillcolor=white](0.25,2.75){0.25}
 \pscircle[fillstyle=solid, fillcolor=white](2.25,2.75){0.25}
 \pscircle[fillstyle=solid, fillcolor=white](1.25,3.75){0.25}
 
\end{pspicture}
\begin{pspicture}(2.5,4)
\psset{unit=0.9cm}
\psline[linecolor=black](0.25,0.75)(1.25,1.75)
\psline[linecolor=black](2.25,0.75)(1.25,1.75)
\psline[linecolor=black](1.25,1.75)(0.25,2.75)
\psline[linecolor=black](1.25,1.75)(2.25,2.75)
\psline[linecolor=black](0.25,2.75)(1.25,3.75)
\psline[linecolor=black](2.25,2.75)(1.25,3.75)
\psline[linecolor=black](2.25,2.75)(3.61,2.75)

 \pscircle[fillstyle=solid, fillcolor=white](0.25,0.75){0.25}
 \pscircle[fillstyle=solid, fillcolor=white](2.25,0.75){0.25}
 \pscircle[fillstyle=solid, fillcolor=white](1.25,1.75){0.25}
 \pscircle[fillstyle=solid, fillcolor=white](3.61,2.75){0.25}
  \pscircle[fillstyle=solid, fillcolor=white](0.25,2.75){0.25}
 \pscircle[fillstyle=solid, fillcolor=white](2.25,2.75){0.25}
 \pscircle[fillstyle=solid, fillcolor=white](1.25,3.75){0.25}
 
\end{pspicture}
\caption{Hence, if the tree induced by the bridges has three leaves, the graph cannot be connected with two colours, as is the case with the three graphs depicted here.}
\end{subfigure}
\caption{}
\label{evencycles}
\end{figure}

 \item[$\Leftarrow$:] Assume that there exists an edge $e$ such that $G+e$ is 2-edge-connected (if $G$ is already 2-edge-connected, any edge $e$ can be used in the rest of the proof, even if $e$ is already in $G$). Let us consider an ear-decomposition $C,P_1,P_2,\dots,P_k$ of $G+e$ such that the cycle
 $C$ uses the edge $e$.
 
 We now build by induction on $i\in [0,k]$ an orientation of $G+e$ such that for every pair of vertices $\{u,v\}$ of $C\cup P_1\cup \dots\cup P_i$, there exists a directed walk from $u$ to $v$ or from $v$ to $u$ that does not use $e$.
 \begin{itemize}
 \item[$\bullet$] We orient the edges of $C$ in such a way that $C$ becomes a directed cycle. Hence, $C$ satisfies the induction hypothesis.
 \item[$\bullet$] For $i\in [1,k]$, let $a$ and $b$ be the extremities of $P_i$ such that there exists a directed path from $a$ to $b$ in $C\cup P_1\cup \dots \cup P_{i-1}$ that does not use $e$. We then orient the edges of $P_i$ as a directed path from $b$ to $a$. The vertex sets of $C\cup P_1\cup\dots\cup P_{i-1}$ and $P_i$ both satisfy the induction hypothesis but it remains to prove that their union also does. Let $u\in P_i$ and $v\in C\cup P_1\cup\dots\cup P_{i-1}$. If there exists a directed walk $W$ from $a$ to $v$ that does not use $e$, then one can use $P_i$ from $u$ to $a$ and $W$ from $a$ to $v$. Otherwise, we know that there exists a directed walk from $v$ to $a$ that does not use $e$. By induction, we know that there also exists one from $a$ to $b$, and we can use $P_i$ to go from $b$ to $u$. The induction hypothesis stands.
  \end{itemize}

 Let $(V_1, V_2)$ be a bi-partition of the vertices of $G$. We now replace all the arcs going from a vertex of $V_1$ to a vertex of $V_2$ by a red edge and all the arcs going from a vertex of $V_2$ to a vertex of $V_1$ by a blue edge. Every directed walk is thus replaced by a properly coloured walk and this 2-edge-colouring connects $G$.\qedhere
 \end{itemize}
 
\end{proof}

 This criteria comes down to checking whether the tree induced by the bridges of the graph is a path, which can  be done in linear time via a depth-first search.

 Note that in the case of bipartite graphs, if two vertices can be connected by a properly coloured walk, then they can also be connected by a properly coloured path. However, as illustrated in Figure \ref{mainexample}, the presence of odd cycles can allow for much more complicated connecting walks.
 
 As illustrated in Figure \ref{oddcyclea}, odd cycles can make it possible to connect graphs that are arbitrarily far from being 2-edge-connected. On the other hand, the graph depicted in Figure \ref{oddcycleb} can be made 2-edge-connected by adding the edge $uv$ and still, no colouring of its edges can make it properly connected. For example, the edge-colouring depicted in Figure \ref{oddcycleb} does not connect the vertices $v$ and $w$. Thus, Theorem \ref{thmbipartite} does not extend to all graphs.

 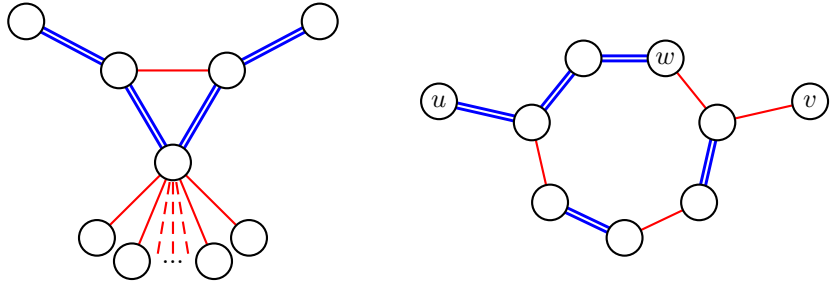
\begin{figure}[!h]
\begin{subfigure}{0.48\linewidth}
\[\begin{pspicture}(2.5,3.7)
\psline[linecolor=red](0.25,0.25)(1.25,1.25)
\psline[linecolor=red](2.25,0.25)(1.25,1.25)
\psline[linecolor=red, linestyle=dashed](1.25,0.04)(1.25,1.25)
\psline[linecolor=red, linestyle=dashed](1.05,0.04)(1.25,1.25)
\psline[linecolor=red, linestyle=dashed](1.45,0.04)(1.25,1.25)
\psline[linecolor=red](0.71,-0.06)(1.25,1.25)
\psline[linecolor=red](1.79,-0.06)(1.25,1.25)
\psline[linecolor=blue](1.25,1.25)(0.54,2.47)
\colorprint{\psline[linecolor=blue, linewidth=3pt](1.25,1.25)(0.54,2.47)
\psline[linecolor=white](1.25,1.25)(0.54,2.47)}
\psline[linecolor=blue](1.25,1.25)(1.96,2.47)
\colorprint{\psline[linecolor=blue, linewidth=3pt](1.25,1.25)(1.96,2.47)
\psline[linecolor=white](1.25,1.25)(1.96,2.47)}
\psline[linecolor=red](0.54,2.47)(1.96,2.47)

\psline[linecolor=blue](-0.68,3.12)(0.54,2.47)
\colorprint{\psline[linecolor=blue, linewidth=3pt](-0.68,3.12)(0.54,2.47)
\psline[linecolor=white](-0.68,3.12)(0.54,2.47)}
\psline[linecolor=blue](3.18,3.12)(1.96,2.47)
\colorprint{\psline[linecolor=blue, linewidth=3pt](3.18,3.12)(1.96,2.47)
\psline[linecolor=white](3.18,3.12)(1.96,2.47)}

 \pscircle[fillstyle=solid, fillcolor=white](0.25,0.25){0.25}
 \pscircle[fillstyle=solid, fillcolor=white](2.25,0.25){0.25}
 \pscircle[fillstyle=solid, fillcolor=white](1.25,1.25){0.25}
 \pscircle[fillstyle=solid, fillcolor=white](0.71,-0.06){0.25}
 \pscircle[fillstyle=solid, fillcolor=white](1.79,-0.06){0.25}
\pscircle[fillstyle=solid, fillcolor=white](0.54,2.47){0.25}
\pscircle[fillstyle=solid, fillcolor=white](1.96,2.47){0.25}
\pscircle[fillstyle=solid, fillcolor=white](3.18,3.12){0.25}
\pscircle[fillstyle=solid, fillcolor=white](-0.68,3.12){0.25}
\rput(1.25,-0.06){...}

\end{pspicture}\]
\caption{A properly connected 2-edge-coloured graph with an arbitrary number of leaves.}
\label{oddcyclea}
\end{subfigure}
\begin{subfigure}{0.48\linewidth}
\[\begin{pspicture}(2.5,3.7)

\psline[linecolor=red](1.25,0.25)(2.23,0.72)
\psline[linecolor=blue](2.23,0.72)(2.47,1.78)
\colorprint{\psline[linecolor=blue, linewidth=3pt](2.23,0.72)(2.47,1.78)
\psline[linecolor=white](2.23,0.72)(2.47,1.78)}
\psline[linecolor=red](2.47,1.78)(1.79,2.63)
\psline[linecolor=blue](1.79,2.63)(0.71,2.63)
\colorprint{\psline[linecolor=blue, linewidth=3pt](1.79,2.63)(0.71,2.63)
\psline[linecolor=white](1.79,2.63)(0.71,2.63)}
\psline[linecolor=blue](0.71,2.63)(0.03,1.78)
\colorprint{\psline[linecolor=blue, linewidth=3pt](0.71,2.63)(0.03,1.78)
\psline[linecolor=white](0.71,2.63)(0.03,1.78)}
\psline[linecolor=red](0.03,1.78)(0.27,0.72)
\psline[linecolor=blue](0.27,0.72)(1.25,0.25)
\colorprint{\psline[linecolor=blue, linewidth=3pt](0.27,0.72)(1.25,0.25)
\psline[linecolor=white](0.27,0.72)(1.25,0.25)}

\psline[linecolor=red](2.47,1.78)(3.69,2.06)
\psline[linecolor=blue](0.03,1.78)(-1.19,2.06)
\colorprint{\psline[linecolor=blue, linewidth=3pt](0.03,1.78)(-1.19,2.06)
\psline[linecolor=white](0.03,1.78)(-1.19,2.06)}

\pscircle[fillstyle=solid, fillcolor=white](1.25,0.25){0.25}
\pscircle[fillstyle=solid, fillcolor=white](2.23,0.72){0.25}
\pscircle[fillstyle=solid, fillcolor=white](2.47,1.78){0.25}
\pscircle[fillstyle=solid, fillcolor=white](1.79,2.63){0.25}
\pscircle[fillstyle=solid, fillcolor=white](0.71,2.63){0.25}
\pscircle[fillstyle=solid, fillcolor=white](0.03,1.78){0.25}
\pscircle[fillstyle=solid, fillcolor=white](0.27,0.72){0.25}

\pscircle[fillstyle=solid, fillcolor=white](3.69,2.06){0.25}
\pscircle[fillstyle=solid, fillcolor=white](-1.19,2.06){0.25}

\rput(-1.19,2.06){$u$}
\rput(3.69,2.06){$v$}
\rput(1.79,2.63){$w$}

\end{pspicture}\]
\caption{No colouring of the edges can connect this graph.}
\label{oddcycleb}
\end{subfigure}
\caption{Counter-examples to the generalization of Theorem \ref{thmbipartite} to non-bipartite graphs.}
\end{figure}

In the next subsections, we study the impact of odd cycles on the connectability of a graph.

\subsection{Stubborn edges and pivots}

We define the \textbf{stubborn edges} of a graph as the edges that belong to every closed walk of odd length. We denote by $\mathscr S$ the set of stubborn edges of a graph. Note that one can check whether a given edge $e$ is stubborn in time $O(n+m)$ by checking whether $G-e$ is bipartite. In the case of a bipartite graph, every closed walk is even and every edge is therefore stubborn.

\begin{prop}\label{propevencycle}
 Let $G$ be a non-bipartite graph. Then, no stubborn edge can appear exactly once in a closed walk of even length in $G$.
\end{prop}

\begin{proof} If an edge $uv$ appears exactly once in an even closed walk $C$, then $C-uv$ is a walk of odd length between $u$ and $v$ that does not use the edge $uv$. Consider now an odd closed walk in $G$ and replace every occurrence of $uv$ by $C-uv$. The resulting walk is still closed, odd and does not use $uv$ which contradicts the stubbornness of $uv$.
\end{proof}

Given a 2-edge-coloured graph, we call a vertex $u$ a \textbf{pivot} if and only if there exists an odd properly coloured closed walk $C$ starting and ending in $u$. Note that $C$ can repeat vertices and edges. Since $C$ is odd and properly coloured, its first and last edges both have the same colour and $C$ cannot be concatenated with itself to yield a new properly coloured walk, unlike in the case of non-edge-coloured graphs.

\vspace{0.35cm}
The following important properties hold.

\begin{prop}
\label{samedirection}
 Let $G$ be a connected graph and $uv\in E(G)$. The edge $uv$ is a stubborn edge if and only if there is no 2-edge-colouring of $G$ such that the edge $uv$ can be used in both directions (to go from $u$ to $v$ and to go from $v$ to $u$) in the same properly coloured walk.
\end{prop}

\begin{proof}
 $\Rightarrow$: let $G$ be 2-edge-coloured and let $W$ be a properly coloured walk that uses $uv$ in both directions. Consider two consecutive occurrences of $uv$ in $W$ with different directions. Because the walk has to alternate colours, there must be an odd number of edges between two occurrences of an edge. Because the edge is used in two different directions, the subwalk of $W$ between the two occurrences is closed, must be odd and cannot use $uv$. Hence, $uv$ is not stubborn.
 
 $\Leftarrow$: let $uv$ be a non-stubborn edge and let $C$ be an odd closed walk that does not use $uv$. Let $P$ be a shortest path between $\{u,v\}$ and $C$. Therefore, $P$ uses neither the edge $uv$ nor an edge of $C$. By symmetry, say that $u$ is the endpoint of $P$ in $\{u,v\}$ and let $w$ be its endpoint in $C$ (it may happen that $u=w$ and $v\in C$ but this does not invalidate the rest of the proof). We give colour 1 to the edge $uv$ and we colour alternatively the path $P$ starting at $u$ with colour 2. We then colour alternatively the edges of $C$ starting and ending on $w$ in such a way that the edges adjacent to $w$ in $C$ are coloured 2 if $P$ has even length and 1 if $P$ has odd length. By concatenating $vu$, $P$ from $u$ to $w$, $C$, $P$ from $w$ to $u$ and $uv$, we form a properly coloured walk that uses $uv$ in both directions.
\end{proof}

\begin{prop}\label{exactlyonce}
 A properly coloured odd closed walk $C$ in a 2-edge-coloured graph uses each stubborn edge exactly once.
 \end{prop}

 \begin{proof}
 Let $C$ be a properly coloured closed walk. By the definition of stubborn edges, $C$ uses each stubborn edge at least once. It remains to prove that it uses them at most once. Suppose that $uv$ is a stubborn edge that appears at least twice in $C$. By Proposition \ref{samedirection}, $uv$ is used in the same direction each time, say from $u$ to $v$. Let $C'$ be the subwalk of $C$ starting with the first occurrence of $uv$ and ending just before the second. Hence, $C'$ is properly coloured (since it is a subwalk of $C$), is closed (starts and ends on $u$) and uses $uv$ exactly once. By Proposition \ref{propevencycle}, $C'$ cannot be even. Hence, its first and last edge both have the same colour $c(uv)$. This leads to a contradiction since $C$ uses the last edge of $C'$ and $uv$ consecutively.
 \end{proof}

    \subsection{$\mathscr S$-free components}
  
Let $G$ be a graph and let $G\setminus\mathscr S$ be the graph obtained from $G$ by removing all the stubborn edges. By an \textbf{$\mathscr S$-free component} of $G$, we mean a connected component of $G\setminus\mathscr S$. Note that if $\mathscr S$ is non-empty, then all $\mathscr S$-free components are bipartite.

\begin{prop}
 Let $G=(V,E)$ be a connected non-bipartite graph and let $\mathcal K$ be an $\mathscr S$-free component of $G$. Then, either $\mathcal K=V(G)$ or there are exactly two edges connecting vertices of $\mathcal K$ to vertices of $V\setminus\mathcal K$.
\end{prop}

\begin{proof}
 In this proof, a \textbf{$\mathcal K$-bridge} is an edge connecting a vertex of $\mathcal K$ to vertex of $V\setminus\mathcal K$. Let us first note that all $\mathcal K$-bridges are stubborn, by definition of $\mathscr S$-free components.

 Let us assume that $\mathcal K\neq V(G)$. Since $G$ is non-bipartite, it contains an elementary odd cycle and every elementary odd cycle has to use all the $\mathcal K$-bridges. Since an odd closed walk can only use a stubborn edge once (by Proposition \ref{exactlyonce}), there must be at least two $\mathcal K$-bridges in the graph. All we have left to prove is that $G$ cannot have strictly more than two $\mathcal K$-bridges. 
 
 Let $C$ be an odd closed walk, let $u\notin \mathcal K$ be a vertex of $C$, let $e_1=v_1u_1$ and $e_2=u_2v_2$ be the $\mathcal K$-bridges right before and after $u$ in $C$ (with $u_1$ and $u_2\notin \mathcal K$) and let us assume that there is another $\mathcal K$-bridge $e_3$ in $G$. Thus, $C$ defines a walk $W_1$ from $u_1$ to $u_2$ that does not use any bridge.  By connectivity of $\mathcal K$ in $G\setminus \mathscr S$, there exists a walk $W_2$ from $v_1$ to $v_2$ that does not use any bridge either. The concatenation of $W_1$, $e_1$, $W_2$ and $e_2$ defines a closed walk $W'$ that must be even since it does not contain $e_3$, but $W'$ uses $e_1$ exactly once, which contradicts Proposition \ref{propevencycle}.
\end{proof}

\begin{prop} If a graph $G$ has several stubborn edges, none of them can have its two endpoints in the same $\mathscr S$-free component.\label{propsevse}
 \end{prop}

\begin{proof}
 Let $uv$ be a stubborn edge such that $u$ and $v$ belong to a same $\mathscr S$-free component. Then there exists a walk from $u$ to $v$ that uses no stubborn edge. This walk forms a closed walk with $uv$ that only contains one stubborn edge and can therefore not be odd (since the graph contains several stubborn edges) but cannot be even by Proposition \ref{propevencycle}, which is a contradiction.
\end{proof}

Note that a consequence of Proposition \ref{propsevse} is that if $G$ has two stubborn edges or more, $G\setminus\mathscr S$ cannot be connected. 

Putting everything together, we have the following:

\begin{theorem}\ 
\begin{itemize}
 \item A connected graph with no or only one stubborn edge consists of exactly one $\mathscr S$-free component.
 \item A connected graph with $k\geqslant 2$ stubborn edges consists of $k$ $\mathscr S$-free components $C_1,\dots,C_k$ and we can label the stubborn edges $e_1,\dots,e_k$ such that for all $i$, $e_i$ connects a vertex of $C_i$ to a vertex of $C_{(i+1)\!\!\mod k}$.
\end{itemize}\label{thmflexcompo}
\end{theorem}

A direct consequence of this theorem is that every pair of odd closed walks use all the stubborn edges in the same direction (up to the choice of the orientation of the walk) and in the same order (up to cyclic permutation).

\vspace{0.35cm}
We now prove a lemma that will be useful for the proof of Proposition \ref{pivotcomponent}. This lemma studies the number of pairs of consecutive edges of the same colour that a walk uses. For example, if a walk $W$ uses the edges $e_1, e_2, e_3, e_4, e_5$ and $e_6$ where $e_4$ and $e_5$ are red and the others are blue, then $W$ contains three pairs of consecutive edges of the same colour ($e_1e_2$, $e_2e_3$ and $e_4e_5$).

 \begin{lemma}\label{lemmaparity}
 Let $C_1$ and $C_2$ be two odd closed walks in a graph $G$, let $e_1=v_1u_1$ and $e_2=u_2v_2$ be two stubborn edges such that $C_1$ and $C_2$ define walks $W_1$ and $W_2$ between $u_1$ and $u_2$ that do not use $e_1$ and $e_2$. Let $c$ be an arbitrary 2-edge-colouring of $G$ and let $n_1$ and $n_2$ be the number of pairs of consecutive edges of same colour in the walks  $e_1W_1e_2$ and $e_1W_2e_2$ respectively. Then, $n_1$ and $n_2$ have same parity.
\end{lemma}

\begin{proof}
Let $\bar{n_1}$ and $\bar{n_2}$ be the number of pairs of consecutive edges of different colours in $e_1W_1e_2$ and $e_1W_2e_2$. We notice that if $e_2$ and $e_1$ have the same colour, then $\bar{n_1}$ and $\bar{n_2}$ are even and that if they have different colours, $\bar{n_1}$ and $\bar{n_2}$ are odd. Hence, $\bar{n_1}$ and $\bar{n_2}$ have same parity.

 Also note that $W_1$ and $W_2$ form a closed walk that does not contain $e_1$ and $e_2$. By stubbornness of $e_1$ and $e_2$, this closed walk is even, which means that the lengths of $W_1$ and $W_2$ have same parity. In other words, $n_1+\bar{n_1}$ and $n_2+\bar{n_2}$ have same parity. The lemma follows.
\end{proof}

\begin{prop}\label{pivotcomponent}
All the pivots of a 2-edge-coloured graph $G$ are in the same $\mathscr S$-free component.
\end{prop}

\begin{proof}
If $G\setminus\mathscr S$ is connected, the proposition immediately follows. Let us now assume that $G\setminus \mathscr S$ is not connected. 

Let $C_1$ and $C_2$ be two properly coloured odd closed walks and let $p_1$ and $p_2$ respectively be their pivots. We denote by $\mathcal K_p$ the $\mathscr S$-free component that contains $p_1$.

Let $e_1=v_1u_1$ and $e_2=u_2v_2$ with $u_1,u_2\in \mathcal K_p$ be the two edges connecting $\mathcal K_p$ to the rest of the graph.
  
Hence, $C_1$ defines two walks between $u_1$ and $u_2$. One of them, let us call it $W_1$, stays within $\mathcal K_p$ and uses no stubborn edge. This walk uses $p_1$ and we know that the two edges adjacent to $p_1$ have the same colour. Hence, this walk contains $n_1=1$ pair of consecutive edges of the same colour. Similarly, $C_2$ defines a walk $W_2$ between $u_1$ and $u_2$ that stays within $\mathcal K_p$ and uses no stubborn edge. By Lemma \ref{lemmaparity}, the number $n_2$ of pairs of consecutive edges of the same colour in $W_2$ is odd. Since $C_2$ is properly coloured, we find that $n_2=1$ and thus, $W_2$ contains $p_2$. 

Hence, if a pivot $p_1$ belongs to an $\mathscr S$-free component $\mathcal K_p$, every other pivot $p_2$ in the graph belongs to $\mathcal K_p$ too.
\end{proof}

 \begin{prop}\label{propCN}
 Let $G$ be a non-bipartite 2-edge-coloured graph. Then, there exists an $\mathscr S$-free component $\mathcal K$ of $G$ such that there is no properly coloured walk between two vertices $u,v\notin \mathcal K$ that goes through a vertex of $\mathcal K$.
\end{prop}

\begin{proof}
 If $G\setminus\mathscr S$ is connected, we are done so assume that $G\setminus\mathscr S$ is disconnected.
 Since $G$ is not bipartite, there exists an odd closed walk $C$ in $G$. Since we only use two colours, we know that $C$ necessarily has an odd number of pairs of consecutive edges of same colour.
 For each such pair, we look at the $\mathscr S$-free component that contains the vertex between the two adjacent edges of the same colour. We thus know that there exists an $\mathscr S$-free component $\mathcal K$ that contains an odd number of occurrences of such vertices of $C$.
Let $e_1=u_1v_1$ and $e_2=u_2v_2$ be the edges between $\mathcal K$ and $V\setminus\mathcal K$, with $u_1$ and $u_2\in\mathcal K$. The walk $C$ defines a walk $W_1$ in $\mathcal K$ from $u_1$ to $u_2$. By the choice of $\mathcal K$, the number $n_1$ of pairs of consecutive edges of the same colour in $e_1W_1e_2$ is odd.
 
 Suppose that $W$ is a properly coloured walk connecting two vertices $u$ and $v\notin\mathcal K$ and going through $w\in\mathcal K$. By symmetry, we assume that $W$ uses $e_1$ to go from $u$ to $w$, and thus, to enter $\mathcal K$. Since $e_1$ is stubborn, by Proposition \ref{samedirection}, $W$ cannot use $e_1$ in the opposite direction and must therefore use $e_2$ to go from $w$ to $v$. This means that $W$ defines a walk $W_2$ from $u_1$ to $u_2$ such that $e_1W_2e_2$ is properly coloured and therefore does not contain any pair of consecutive edges of the same colour. However, by Lemma \ref{lemmaparity}, the walk $e_1W_2e_2$ must contain an odd number of pair of adjacent edges of the same colour, which is a contradiction.
 \end{proof}

 \subsection{Connecting graphs with no stubborn edge}

 In \cite{MelvilleGoddard}, Melville and Goddard proved that if a graph contains two edge-disjoint odd cycles, then it can always be connected with two colours. Studying the stubborn edges leads to a nice generalization of this result and we prove in this subsection that two colours actually suffice to connect any graph that contains no stubborn edge (Theorem \ref{thmnostubborn}). We first need to study the structure of such graphs.

\begin{prop}\label{prop3cycles}
 If a graph $G$ has no stubborn edge, then we can find in polynomial time a set of at most three odd cycles such that
  no edge belongs to all of them.
\end{prop}

\begin{proof}
Since $G$ contains no stubborn edge, we know that it contains at least two elementary odd cycles.
If there are two cycles with no common edge, the property immediately holds.
Otherwise, let $C_1$ and $C_2$ be two elementary odd cycles $C_1$ and $C_2$ and let $P=v_1\dots v_k$ be the shortest subpath of $C_1$ that contains all the edges of $C_1\cap C_2$. Hence, the first and last edges of $P$ belong to $C_2$ too. Let $e_i=v_iv_{i+1}$ for $1\leqslant i \leqslant k$.

Let $C_3$ be an elementary odd cycle that does not use $e_1$. We know that such a cycle exists because $e_1$ is not stubborn.
If $C_3$ does not use any edge of $C_1\cap C_2$ (which is notably the case if $k=1$), then $C_1\cap C_2\cap C_3$ is empty and we are done. Else, let $e_i=v_iv_{i+1}$ be the first edge of $P$ that belongs to $C_3$ (hence, $2\leq i\leq k$). If all the edges of $C_3$ that belong to $C_1$ or $C_2$ are in $P$, this means that $C_1$ and $C_3$ are two elementary cycles whose intersection is contained in a subpath of $C_1$ strictly shorter than $P$. We iterate this process with $C_1$ and $C_3$ instead of $C_1$ and $C_2$. Else, let $e$ be the last edge of $C_1\cup C_2\setminus P$ that appears in $C_3$ before $C_3$ uses $e_i$. 
Thus, $e$ has an endpoint $u$ such that $C_3$ defines a walk $W_3$ between $u$ and $v_i$ that uses no edge of $C_1$ or $C_2$. If $u\in C_1$, then $C_1$ defines two walks $W_1$ and $W_2$ of different parity from $u$ to $v_i$. One of these walks forms with $W_3$ an odd cycle $C_4$ whose intersection with $C_2$ is contained in a subpath of $C_4$ strictly shorter than $P$. Similarly, if $u\in C_2$, then $C_2$ defines two walks of different parity from $u$ to $v_i$ and one of these walks forms with $W_3$ an odd cycle $C_4$ whose intersection with $C_1$ is contained in a subpath of $C_1$ strictly shorter than $P$.

Hence, we can iterate this process with $C_4$ and $C_2$ or $C_1$ and $C_4$ instead of $C_1$ and $C_2$. We know that it eventually ends since $P$ is shorter at each iteration and the process necessarily ends if $P$ has length one.
\end{proof} 

Note that the proof above can be turned into a polynomial time algorithm to build three odd cycles $C_1$, $C_2$ and $C_3$ such that no edge belong to all three of them. By iterating the above algorithm as long as we can find an odd cycle $C_3$ that does not use $e_1$ or $e_{k-1}$ but still shares edges with $P$, we can ensure that the intersection between $C_1$ and $C_2$ is a path $P$ minimal by inclusion \textit{i.e.} such that no two odd cycles intersect in a proper subpath of $P$.

\vspace{0.35cm}
Again, the following theorem has appeared independently in \cite{MelvilleGoddard} but we leave it in our paper because the construction we use in its proof will serve as a basis in the proof of Theorem \ref{thmnostubborn}.

\begin{theorem}\label{2disjcycles}
 If a connected graph has two edge-disjoint odd cycles, then it can be connected with two colours.
\end{theorem}

\begin{proof}
 Let $G$ be a connected graph and let $C_1$ and $C_2$ be two edge-disjoint odd cycles in $G$. Let $P$ be a path connecting a vertex $u\in C_1$ to a vertex $v\in C_2$ such that no intermediate vertex of $P$ belongs to either $C_1$ or $C_2$.

 Note that the shortest path from any vertex of $G$ to $C_1\cup P\cup C_2$  uses no edge of $C_1$, $P$ or $C_2$. Furthermore, all the vertices of $C_1$ can be reached from $u$ by walks that only use edges of $C_1$ and use therefore no edge of $P$ or $C_2$. Similarly, all the vertices of $P$ or $C_2$ can be reached from $u$ without using any edge from $C_1$. Hence, every vertex of the graph can be reached from $u$ either without using any edge of $C_1$ or without using any edge of $P$ and $C_2$. 
 
 We denote by $V_1$ the set of vertices of $G$ that can be reached from $u$ without using the edges of $P$ or $C_2$ and by $V_2$ be its complement (those vertices can thus be reached from $u$ without using the edges of $C_1$). Consider for example the graph depicted in Figure \ref{fig2dc1}, let $C_1=(v_1,v_3,v_4,v_5,v_2,v_1)$, $C_2=(v_8,v_{11},v_9,v_8)$ and $P=(v_4,v_8)$. Here, we have $u=v_4$, $V_1=\{v_1,v_2,v_3,v_4,v_5,v_6\}$ and $V_2=\{v_7,v_8,v_9,v_{10},v_{11},v_{12}\}$.
 
 We create the graph $H$ from $G$ by removing all the edges of $P$ and $C_2$. We then use colour 1 on all the edges that are at even distance of $u$ in $H$ and colour 2 and all the edges at odd distance. We do not colour the edges that cannot be reached from $u$ in $H$. This first step is illustrated in Figure \ref{fig2dc2} (where red is colour 1 and blue is colour 2). 
 By doing so, we colour all the edges except those of $P$ and $C_2$ (that are not in $H$) and those whose end vertices are in $V_2$ (that cannot be reached from $u$ in $H$). After this step, it is possible that $C_1$ is not properly coloured (as is the case in Figure \ref{fig2dc2}) but we can still prove that there exists an odd properly coloured closed walk $W_1$ with pivots $u$ in $H$. Indeed, we create properly coloured shortest walks from $u$ to any vertex of $V_1$ starting with colour 1 and ending with a colour that depends on the length of the walk. Moreover, if two vertices are adjacent, their distance from $u$ can either be the same or differ by 1. By comparing the distance from $u$ of all the pair of adjacent vertices of $C_1$, we find that there must be two adjacent vertices $w$ and $w'$ at same distance from $u$ since $C_1$ is closed and odd. Thus, the concatenation of a shortest walk from $u$ to $w$, the edge between $w$ and $w'$ and a shortest walk from $w'$ to $u$ is properly coloured too. For example, in Figure \ref{fig2dc2}, we see that the vertices $v_1$ and $v_3$ are both at the same distance from $u=v_4$. We can thus construct the properly coloured closed walk $W_1=(v_4,v_1,v_3,v_4)$ in $H$.
 
 We then create a copy $H'$ of $G$ by removing all the already coloured edges. We use colour 2 on all the edges at even distance from $u$ in $H'$ and colour 1 on all the edges at odd distance from $u$. This second step is illustrated in Figure \ref{fig2dc3}. Just like before, this step creates an odd properly coloured closed walk $W_2$ with pivot $u$ in $C_2$. Indeed, there always exists two adjacent vertices in $C_2$ that are at the same distance from $u$. In Figure \ref{fig2dc3}, the walk $W_2=(v_4,v_8,v_{11}v_9,v_4)$ is odd and properly coloured.
 
 We claim that the resulting colouring connects the graph. Indeed, the first step 
 creates properly coloured walks from every vertex of $V_1$ to $u$ that ends on an edge coloured 1 and the second creates properly coloured walks from $u$ to every vertex of $V_2$ that starts with an edge coloured 2. Furthermore, $W_1$ goes from $u$ to $u$ starting and ending on an edge coloured 1 and $W_2$ goes from $u$ to $u$ starting and ending with an edge coloured 2. The walks of the first steps together with $W_2$ allow to connect any two vertices of $V_1$. For example, in Figure \ref{fig2dc4}, one can go from $v_2$ to $v_3$ by using $(v_2, v_1,v_4)$ to reach $u=v_4$, use $W_2=(v_4,v_8,v_{11},v_9,v_8,v_4)$ and finally, $(v_4,v_3)$ to reach $v_3$. Similarly, the walks of the second step together with $W_1$ connect any two vertices of $V_2$. Finally, the two steps together connect the vertices of $V_1$ with those of $V_2$. For example, in Figure \ref{fig2dc4}, one can go from $v_6$ to $v_7$ by going from $v_6$ to $u=v_4$ with the walk $(v_6,v_3,v_4)$ and then use $(v_4,v_8,v_9,v_{10},v_7)$ to reach $v_7$.\qedhere

\begin{figure}[!h]
\begin{subfigure}{0.48\linewidth}
\[\begin{pspicture}(3.96,2.5)

\psline(3.71,0.75)(3.71,1.75)
\psline(3.71,0.75)(2.85,0.25)
\psline[linewidth=2pt](3.71,1.75)(2.85,1.25)
\psline[linewidth=2pt](3.71,1.75)(2.85,2.25)
\psline(2.85,0.25)(2.85,1.25)
\psline(2.85,0.25)(1.98,0.75)
\psline[linewidth=2pt](2.85,1.25)(2.85,2.25)
\psline[linewidth=2pt](2.85,2.25)(1.12,1.25)
\psline(1.98,2.25)(1.12,2.25)
\psline[linewidth=2pt](1.12,0.25)(1.12,1.25)
\psline[linewidth=2pt](1.12,0.25)(0.25,0.75)
\psline[linewidth=2pt](1.12,1.25)(1.12,2.25)
\psline(1.12,1.25)(0.25,1.75)
\psline[linewidth=2pt](1.12,2.25)(0.25,1.75)
\psline[linewidth=2pt](0.25,0.75)(0.25,1.75)

\pscircle[fillstyle=solid, fillcolor=white](3.71,0.75){0.25}
\pscircle[fillstyle=solid, fillcolor=white](3.71,1.75){0.25}
\pscircle[fillstyle=solid, fillcolor=white](2.85,0.25){0.25}
\pscircle[fillstyle=solid, fillcolor=white](2.85,1.25){0.25}
\pscircle[fillstyle=solid, fillcolor=white](2.85,2.25){0.25}
\pscircle[fillstyle=solid, fillcolor=white](1.98,0.75){0.25}
\pscircle[fillstyle=solid, fillcolor=white](1.98,2.25){0.25}
\pscircle[fillstyle=solid, fillcolor=white](1.12,0.25){0.25}
\pscircle[fillstyle=solid, fillcolor=white](1.12,1.25){0.25}
\pscircle[fillstyle=solid, fillcolor=white](1.12,2.25){0.25}
\pscircle[fillstyle=solid, fillcolor=white](0.25,0.75){0.25}
\pscircle[fillstyle=solid, fillcolor=white](0.25,1.75){0.25}

\rput(0.25,0.75){$v_2$}
\rput(0.25,1.75){$v_1$}
\rput(1.12,0.25){$v_5$}
\rput(1.12,1.25){$v_4$}
\rput(1.12,2.25){$v_3$}
\rput(1.98,0.75){$v_7$}
\rput(1.98,2.25){$v_6$}
\rput(2.85,0.25){$v_{10}$}
\rput(2.85,1.25){$v_9$}
\rput(2.85,2.25){$v_8$}
\rput(3.71,0.75){$v_{12}$}
\rput(3.71,1.75){$v_{11}$}
\rput(3.13,1.75){$C_2$}
\rput(0.56,1.25){$C_1$}
\rput(2.01,1.5){$P$}

\end{pspicture}\]
\caption{An example of graph with two edge-disjoint odd cycles.}
\label{fig2dc1}
\end{subfigure}
\begin{subfigure}{0.48\linewidth}
\[\begin{pspicture}(3.96,2.5)

\psline(3.71,0.75)(3.71,1.75)
\psline(3.71,0.75)(2.85,0.25)
\psline[linecolor=gray, linestyle=dashed](3.71,1.75)(2.85,1.25)
\psline[linecolor=gray, linestyle=dashed](3.71,1.75)(2.85,2.25)
\psline(2.85,0.25)(2.85,1.25)
\psline(2.85,0.25)(1.98,0.75)
\psline[linecolor=gray, linestyle=dashed](2.85,1.25)(2.85,2.25)
\psline[linecolor=gray, linestyle=dashed](2.85,2.25)(1.12,1.25)
\psline[linecolor=blue](1.98,2.25)(1.12,2.25)
\colorprint{\psline[linecolor=blue, linewidth=3pt](1.98,2.25)(1.12,2.25)
\psline[linecolor=white](1.98,2.25)(1.12,2.25)}
\psline[linecolor=red](1.12,0.25)(1.12,1.25)
\psline[linecolor=blue](1.12,0.25)(0.25,0.75)
\colorprint{\psline[linecolor=blue, linewidth=3pt](1.12,0.25)(0.25,0.75)
\psline[linecolor=white](1.12,0.25)(0.25,0.75)}
\psline[linecolor=red](1.12,1.25)(1.12,2.25)
\psline[linecolor=red](1.12,1.25)(0.25,1.75)
\psline[linecolor=blue](1.12,2.25)(0.25,1.75)
\colorprint{\psline[linecolor=blue, linewidth=3pt](1.12,2.25)(0.25,1.75)
\psline[linecolor=white](1.12,2.25)(0.25,1.75)}
\psline[linecolor=blue](0.25,0.75)(0.25,1.75)
\colorprint{\psline[linecolor=blue, linewidth=3pt](0.25,0.75)(0.25,1.75)
\psline[linecolor=white](0.25,0.75)(0.25,1.75)}

\pscircle[fillstyle=solid, fillcolor=white](0.25,0.75){0.25}
\pscircle[fillstyle=solid, fillcolor=white](0.25,1.75){0.25}
\pscircle[fillstyle=solid, fillcolor=white](1.12,0.25){0.25}
\pscircle[fillstyle=solid, fillcolor=white](1.12,1.25){0.25}
\pscircle[fillstyle=solid, fillcolor=white](1.12,2.25){0.25}
\pscircle[fillstyle=solid, fillcolor=white](1.98,0.75){0.25}
\pscircle[fillstyle=solid, fillcolor=white](1.98,2.25){0.25}
\pscircle[fillstyle=solid, fillcolor=white](2.85,0.25){0.25}
\pscircle[fillstyle=solid, fillcolor=white](2.85,1.25){0.25}
\pscircle[fillstyle=solid, fillcolor=white](2.85,2.25){0.25}
\pscircle[fillstyle=solid, fillcolor=white](3.71,0.75){0.25}
\pscircle[fillstyle=solid, fillcolor=white](3.71,1.75){0.25}

\rput(0.25,0.75){$v_2$}
\rput(0.25,1.75){$v_1$}
\rput(1.12,0.25){$v_5$}
\rput(1.12,1.25){$v_4$}
\rput(1.12,2.25){$v_3$}
\rput(1.98,0.75){$v_7$}
\rput(1.98,2.25){$v_6$}
\rput(2.85,0.25){$v_{10}$}
\rput(2.85,1.25){$v_9$}
\rput(2.85,2.25){$v_8$}
\rput(3.71,0.75){$v_{12}$}
\rput(3.71,1.75){$v_{11}$}

\end{pspicture}\]
\caption{The first step of the colouring algorithm.}
\label{fig2dc2}
\end{subfigure}

\begin{subfigure}{0.48\linewidth}
\[\begin{pspicture}(3.96,2.75)

\psline[linecolor=blue](3.71,0.75)(3.71,1.75)
\colorprint{\psline[linecolor=blue, linewidth=3pt](3.71,0.75)(3.71,1.75)
\psline[linecolor=white](3.71,0.75)(3.71,1.75)}
\psline[linecolor=red](3.71,0.75)(2.85,0.25)
\psline[linecolor=blue](3.71,1.75)(2.85,1.25)
\colorprint{\psline[linecolor=blue, linewidth=3pt](3.71,1.75)(2.85,1.25)
\psline[linecolor=white](3.71,1.75)(2.85,1.25)}
\psline[linecolor=red](3.71,1.75)(2.85,2.25)
\psline[linecolor=blue](2.85,0.25)(2.85,1.25)
\colorprint{\psline[linecolor=blue, linewidth=3pt](2.85,0.25)(2.85,1.25)
\psline[linecolor=white](2.85,0.25)(2.85,1.25)}
\psline[linecolor=red](2.85,0.25)(1.98,0.75)
\psline[linecolor=red](2.85,1.25)(2.85,2.25)
\psline[linecolor=blue](2.85,2.25)(1.12,1.25)
\colorprint{\psline[linecolor=blue, linewidth=3pt](2.85,2.25)(1.12,1.25)
\psline[linecolor=white](2.85,2.25)(1.12,1.25)}
\psline[linecolor=gray, linestyle=dashed](1.98,2.25)(1.12,2.25)
\psline[linecolor=gray, linestyle=dashed](1.12,0.25)(1.12,1.25)
\psline[linecolor=gray, linestyle=dashed](1.12,0.25)(0.25,0.75)
\psline[linecolor=gray, linestyle=dashed](1.12,1.25)(1.12,2.25)
\psline[linecolor=gray, linestyle=dashed](1.12,1.25)(0.25,1.75)
\psline[linecolor=gray, linestyle=dashed](1.12,2.25)(0.25,1.75)
\psline[linecolor=gray, linestyle=dashed](0.25,0.75)(0.25,1.75)

\pscircle[fillstyle=solid, fillcolor=white](0.25,0.75){0.25}
\pscircle[fillstyle=solid, fillcolor=white](0.25,1.75){0.25}
\pscircle[fillstyle=solid, fillcolor=white](1.12,0.25){0.25}
\pscircle[fillstyle=solid, fillcolor=white](1.12,1.25){0.25}
\pscircle[fillstyle=solid, fillcolor=white](1.12,2.25){0.25}
\pscircle[fillstyle=solid, fillcolor=white](1.98,0.75){0.25}
\pscircle[fillstyle=solid, fillcolor=white](1.98,2.25){0.25}
\pscircle[fillstyle=solid, fillcolor=white](2.85,0.25){0.25}
\pscircle[fillstyle=solid, fillcolor=white](2.85,1.25){0.25}
\pscircle[fillstyle=solid, fillcolor=white](2.85,2.25){0.25}
\pscircle[fillstyle=solid, fillcolor=white](3.71,0.75){0.25}
\pscircle[fillstyle=solid, fillcolor=white](3.71,1.75){0.25}

\rput(0.25,0.75){$v_2$}
\rput(0.25,1.75){$v_1$}
\rput(1.12,0.25){$v_5$}
\rput(1.12,1.25){$v_4$}
\rput(1.12,2.25){$v_3$}
\rput(1.98,0.75){$v_7$}
\rput(1.98,2.25){$v_6$}
\rput(2.85,0.25){$v_{10}$}
\rput(2.85,1.25){$v_9$}
\rput(2.85,2.25){$v_8$}
\rput(3.71,0.75){$v_{12}$}
\rput(3.71,1.75){$v_{11}$}

\end{pspicture}\]
\caption{The second step of the algorithm.}
\label{fig2dc3}
\end{subfigure}
\begin{subfigure}{0.48\linewidth}
\[\begin{pspicture}(3.96,2.75)

\psline[linecolor=blue](3.71,0.75)(3.71,1.75)
\colorprint{\psline[linecolor=blue, linewidth=3pt](3.71,0.75)(3.71,1.75)
\psline[linecolor=white](3.71,0.75)(3.71,1.75)}
\psline[linecolor=red](3.71,0.75)(2.85,0.25)
\psline[linecolor=blue](3.71,1.75)(2.85,1.25)
\colorprint{\psline[linecolor=blue, linewidth=3pt](3.71,1.75)(2.85,1.25)
\psline[linecolor=white](3.71,1.75)(2.85,1.25)}
\psline[linecolor=red](3.71,1.75)(2.85,2.25)
\psline[linecolor=blue](2.85,0.25)(2.85,1.25)
\colorprint{\psline[linecolor=blue, linewidth=3pt](2.85,0.25)(2.85,1.25)
\psline[linecolor=white](2.85,0.25)(2.85,1.25)}
\psline[linecolor=red](2.85,0.25)(1.98,0.75)
\psline[linecolor=red](2.85,1.25)(2.85,2.25)
\psline[linecolor=blue](2.85,2.25)(1.12,1.25)
\colorprint{\psline[linecolor=blue, linewidth=3pt](2.85,2.25)(1.12,1.25)
\psline[linecolor=white](2.85,2.25)(1.12,1.25)}
\psline[linecolor=blue](1.98,2.25)(1.12,2.25)
\colorprint{\psline[linecolor=blue, linewidth=3pt](1.98,2.25)(1.12,2.25)
\psline[linecolor=white](1.98,2.25)(1.12,2.25)}
\psline[linecolor=red](1.12,0.25)(1.12,1.25)
\psline[linecolor=blue](1.12,0.25)(0.25,0.75)
\colorprint{\psline[linecolor=blue, linewidth=3pt](1.12,0.25)(0.25,0.75)
\psline[linecolor=white](1.12,0.25)(0.25,0.75)}
\psline[linecolor=red](1.12,1.25)(1.12,2.25)
\psline[linecolor=red](1.12,1.25)(0.25,1.75)
\psline[linecolor=blue](1.12,2.25)(0.25,1.75)
\colorprint{\psline[linecolor=blue, linewidth=3pt](1.12,2.25)(0.25,1.75)
\psline[linecolor=white](1.12,2.25)(0.25,1.75)}
\psline[linecolor=blue](0.25,0.75)(0.25,1.75)
\colorprint{\psline[linecolor=blue, linewidth=3pt](0.25,0.75)(0.25,1.75)
\psline[linecolor=white](0.25,0.75)(0.25,1.75)}

\pscircle[fillstyle=solid, fillcolor=white](0.25,0.75){0.25}
\pscircle[fillstyle=solid, fillcolor=white](0.25,1.75){0.25}
\pscircle[fillstyle=solid, fillcolor=white](1.12,0.25){0.25}
\pscircle[fillstyle=solid, fillcolor=white](1.12,1.25){0.25}
\pscircle[fillstyle=solid, fillcolor=white](1.12,2.25){0.25}
\pscircle[fillstyle=solid, fillcolor=white](1.98,0.75){0.25}
\pscircle[fillstyle=solid, fillcolor=white](1.98,2.25){0.25}
\pscircle[fillstyle=solid, fillcolor=white](2.85,0.25){0.25}
\pscircle[fillstyle=solid, fillcolor=white](2.85,1.25){0.25}
\pscircle[fillstyle=solid, fillcolor=white](2.85,2.25){0.25}
\pscircle[fillstyle=solid, fillcolor=white](3.71,0.75){0.25}
\pscircle[fillstyle=solid, fillcolor=white](3.71,1.75){0.25}

\rput(0.25,0.75){$v_2$}
\rput(0.25,1.75){$v_1$}
\rput(1.12,0.25){$v_5$}
\rput(1.12,1.25){$v_4$}
\rput(1.12,2.25){$v_3$}
\rput(1.98,0.75){$v_7$}
\rput(1.98,2.25){$v_6$}
\rput(2.85,0.25){$v_{10}$}
\rput(2.85,1.25){$v_9$}
\rput(2.85,2.25){$v_8$}
\rput(3.71,0.75){$v_{12}$}
\rput(3.71,1.75){$v_{11}$}

\end{pspicture}\]
\caption{The resulting colouring.}
\label{fig2dc4}
\end{subfigure}
\caption{An example of how to construct a connecting 2-edge-colouring of a graph with two edge-disjoint odd cycles.}
\end{figure}
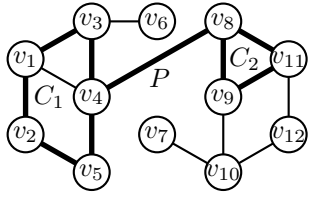
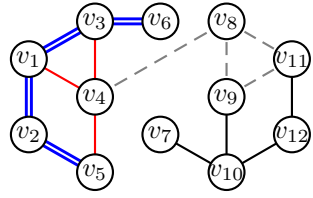
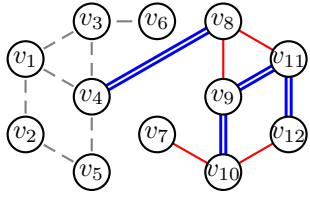
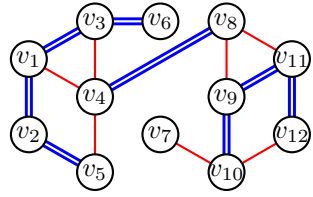
\end{proof}

\begin{theorem}\label{thmnostubborn}
 If a connected graph has no stubborn edge, then it can be connected with two colours.
\end{theorem}

\begin{proof}
 If there are two odd cycles with no edge in common, the claim immediately follows from Theorem \ref{2disjcycles}. We assume in the rest of the proof that this is not the case. 
 
 The proof of Proposition \ref{prop3cycles} provides two odd cycles $C_1$ and $C_2$ whose intersection is a minimal path $Q$ (which means that no two odd cycles intersect in a proper subpath of $Q$), and an odd cycle $C_3$ such that no edge belongs to $C_1$, $C_2$ and $C_3$. By merging $C_1$ and $C_2$ and removing $Q$, we create an elementary even cycle $C_4$ whose intersection with $C_3$ is non-empty since $C_3$ intersects $C_1$ and $C_2$ but not $Q$. Hence, $C_3$ can be decomposed as the concatenation of paths $Q_1, P_1, Q_2, P_2,...$ such that the $Q_i$ are subpaths of $C_4$ too and the $P_i$ use no edge of $C_4$. Let $u_i$ and $v_i$ be the end vertices of $P_i$. Note that for all $i$, $u_i$ and $v_i$ both belong to $C_4$. Hence, $C_4$ defines two paths between $u_i$ and $v_i$ but since $C_4$ is even, they both have same parity. Since $C_3$ is odd and $C_4$ is even, we know that there exists $i$ such that $P_i$ has different parity from the walks $W_i$ and $W'_i$ that $C_4$ defines between $u_i$ and $v_i$. Hence, the concatenation of $P_i$ and $W_i$ and of $P_i$ and $W'_i$ are two odd cycles that we call respectively $\mathscr C$ and $\mathscr C'$. These cycles will be useful later in the proof as they provide paths of different parity between their vertices. Since we assumed that the graph does not contain two edge-disjoint odd cycles, $\mathscr C$ and $\mathscr C'$ must share edges with both $C_1$ and $C_2$. Hence, one of $u_i$ and $v_i$ belongs to $C_1$ and the other belongs to $C_2$. By symmetry, we may assume that $u_i\in C_1$ and $v_i\in C_2$. 
 
 Consider for example the graph depicted in Figure \ref{fignostub}, that has no stubborn edge. Let $C_1=(v_1,v_2,v_3,v_4,v_5,v_6,v_7,v_1)$ and let $C_2=(v_1,v_8,v_9,v_{10},v_{11},v_3,$ $v_2,v_1)$. Those two cycles intersect in a path $Q=(v_1,v_2,v_3)$. The even cycle $C_4$ we create from $C_1$ and $C_2$ is $(v_1,v_8,v_9,v_{10},v_{11},v_3,v_4,v_5,v_6,v_7,v_1)$.
 Let us consider the odd cycle $C_3=(v_4,v_5,v_6,v_{12},v_9,v_{10},v_{11},v_4)$ that does not use any edge of $Q$.  Following the notation of the proof, we decompose it as the concatenation of $Q_1=(v_4,v_5,v_6)$, $P_1=(v_6,v_{12},v_9)$, $Q_2=(v_9,v_{10},v_{11})$ and $P_2=(v_{11},v_4)$. As expected, there exists $i$ (here, $i=2$) such that the walks defined by $C_4$ between the end vertices of $P_i$ do not have the same parity as $P_i$. We can see that $v_4$ belongs to $C_1$ and $v_{11}$ belongs to $C_2$. The even cycle $C_4$ defines two path between $v_4$ and $v_{11}$ and those paths together with $P_i$ creates two odd cycles $\mathscr C$ and $\mathscr C'$: $(v_4,v_{11},v_3,v_4)$ and $(v_4,v_{11},v_{10},v_9,v_8,v_1,v_7,v_6,v_5,v_4)$.

    \begin{figure}[!h]
   \begin{subfigure}{0.48\linewidth}
\[\begin{pspicture}(0,-0.95)(5.5,3.96)
\psline(0.92,0.63)(2.25,-0.13)
\psline(3.58,0.63)(3.58,2.17)
\psline(2.25,2.94)(0.92,2.17)
\psline(0.25,0.25)(2.25,-0.9)
\psline(4.25,0.25)(2.25,-0.9)
\psline(2.25,3.71)(0.25,2.56)
\psline(0.25,0.25)(0.25,2.56)
\psline(4.25,0.25)(4.25,2.56)
\psline(2.25,3.71)(4.25,2.56)
\psline(0.25,0.25)(1.58,1.02)
\psline(2.25,3.71)(2.25,2.17)
\psline(4.25,0.25)(2.92,1.02)
\psline(1.58,1.02)(2.25,2.17)
\psline(1.58,1.02)(2.92,1.02)
\psline(2.25,2.17)(2.92,1.02)

\pscircle[linecolor=black, fillstyle=solid, fillcolor=white](0.25,0.25){0.25}
\pscircle[linecolor=black, fillstyle=solid, fillcolor=white](2.25,3.71){0.25}
\pscircle[linecolor=black, fillstyle=solid, fillcolor=white](4.25,0.25){0.25}
\pscircle[linecolor=black, fillstyle=solid, fillcolor=white](0.92,0.63){0.25}
\pscircle[linecolor=black, fillstyle=solid, fillcolor=white](3.58,0.63){0.25}
\pscircle[linecolor=black, fillstyle=solid, fillcolor=white](2.25,2.94){0.25}
\pscircle[linecolor=black, fillstyle=solid, fillcolor=white](1.58,1.02){0.25}
\pscircle[linecolor=black, fillstyle=solid, fillcolor=white](2.92,1.02){0.25}
\pscircle[linecolor=black, fillstyle=solid, fillcolor=white](2.25,2.17){0.25}
\pscircle[linecolor=black, fillstyle=solid, fillcolor=white](2.25,-0.13){0.25}
\pscircle[linecolor=black, fillstyle=solid, fillcolor=white](0.92,2.17){0.25}
\pscircle[linecolor=black, fillstyle=solid, fillcolor=white](3.58,2.17){0.25}
\pscircle[linecolor=black, fillstyle=solid, fillcolor=white](2.25,-0.9){0.25}
\pscircle[linecolor=black, fillstyle=solid, fillcolor=white](0.25,2.56){0.25}
\pscircle[linecolor=black, fillstyle=solid, fillcolor=white](4.25,2.56){0.25}

\rput(0.25,0.25){$v_6$}
\rput(2.25,3.71){$v_1$}
\rput(4.25,0.25){$v_9$}
\rput(0.92,0.63){$v_5$}
\rput(3.58,0.63){$v_{10}$}
\rput(2.25,2.94){$v_2$}
\rput(1.58,1.02){$v_4$}
\rput(2.92,1.02){$v_{11}$}
\rput(2.25,2.17){$v_3$}
\rput(2.25,-0.13){$v_{14}$}
\rput(0.92,2.17){$v_{13}$}
\rput(3.58,2.17){$v_{15}$}
\rput(2.25,-0.9){$v_{12}$}
\rput(0.25,2.56){$v_7$}
\rput(4.25,2.56){$v_8$}

\end{pspicture}\]
\caption{An example of graph with no \\ stubborn edge.}
\label{fignostub}
\end{subfigure}
   \begin{subfigure}{0.48\linewidth}
\[\begin{pspicture}(0,-0.95)(5.5,3.96)
\psline[linecolor=blue](0.92,0.63)(2.25,-0.13)
\colorprint{\psline[linecolor=blue, linewidth=3pt](0.92,0.63)(2.25,-0.13)
\psline[linecolor=white](0.92,0.63)(2.25,-0.13)}
\psline(3.58,0.63)(3.58,2.17)
\psline[linecolor=red](2.25,2.94)(0.92,2.17)
\psline[linecolor=red](0.25,0.25)(2.25,-0.9)
\psline[linecolor=blue](4.25,0.25)(2.25,-0.9)
\colorprint{\psline[linecolor=blue, linewidth=3pt](4.25,0.25)(2.25,-0.9)
\psline[linecolor=white](4.25,0.25)(2.25,-0.9)}
\psline[linecolor=blue](2.25,3.71)(0.25,2.56)
\colorprint{\psline[linecolor=blue, linewidth=3pt](2.25,3.71)(0.25,2.56)
\psline[linecolor=white](2.25,3.71)(0.25,2.56)}
\psline[linecolor=red](0.25,0.25)(0.25,2.56)
\psline[linecolor=gray, linestyle=dashed](4.25,0.25)(4.25,2.56)
\psline[linecolor=gray, linestyle=dashed](2.25,3.71)(4.25,2.56)
\psline[linecolor=blue](0.25,0.25)(0.92,0.63)
\colorprint{\psline[linecolor=blue, linewidth=3pt](0.25,0.25)(0.92,0.63)
\psline[linecolor=white](0.25,0.25)(0.92,0.63)}
\psline[linecolor=red](0.92,0.63)(1.58,1.02)
\psline[linecolor=red](2.25,3.71)(2.25,2.94)
\psline[linecolor=blue](2.25,2.17)(2.25,2.94)
\colorprint{\psline[linecolor=blue, linewidth=3pt](2.25,2.17)(2.25,2.94)
\psline[linecolor=white](2.25,2.17)(2.25,2.94)}
\psline[linecolor=gray, linestyle=dashed](4.25,0.25)(3.58,0.63)
\psline[linecolor=gray, linestyle=dashed](2.92,1.02)(3.58,0.63)
\psline[linecolor=red](1.58,1.02)(2.25,2.17)
\psline[linecolor=gray, linestyle=dashed](1.58,1.02)(2.92,1.02)
\psline[linecolor=gray, linestyle=dashed](2.25,2.17)(2.92,1.02)

\pscircle[linecolor=black, fillstyle=solid, fillcolor=white](0.25,0.25){0.25}
\pscircle[linecolor=black, fillstyle=solid, fillcolor=white](2.25,3.71){0.25}
\pscircle[linecolor=black, fillstyle=solid, fillcolor=white](4.25,0.25){0.25}
\pscircle[linecolor=black, fillstyle=solid, fillcolor=white](0.92,0.63){0.25}
\pscircle[linecolor=black, fillstyle=solid, fillcolor=white](3.58,0.63){0.25}
\pscircle[linecolor=black, fillstyle=solid, fillcolor=white](2.25,2.94){0.25}
\pscircle[linecolor=black, fillstyle=solid, fillcolor=white](1.58,1.02){0.25}
\pscircle[linecolor=black, fillstyle=solid, fillcolor=white](2.92,1.02){0.25}
\pscircle[linecolor=black, fillstyle=solid, fillcolor=white](2.25,2.17){0.25}
\pscircle[linecolor=black, fillstyle=solid, fillcolor=white](2.25,-0.13){0.25}
\pscircle[linecolor=black, fillstyle=solid, fillcolor=white](0.92,2.17){0.25}
\pscircle[linecolor=black, fillstyle=solid, fillcolor=white](3.58,2.17){0.25}
\pscircle[linecolor=black, fillstyle=solid, fillcolor=white](2.25,-0.9){0.25}
\pscircle[linecolor=black, fillstyle=solid, fillcolor=white](0.25,2.56){0.25}
\pscircle[linecolor=black, fillstyle=solid, fillcolor=white](4.25,2.56){0.25}

\rput(0.25,0.25){$v_6$}
\rput(2.25,3.71){$v_1$}
\rput(4.25,0.25){$v_9$}
\rput(0.92,0.63){$v_5$}
\rput(3.58,0.63){$v_{10}$}
\rput(2.25,2.94){$v_2$}
\rput(1.58,1.02){$v_4$}
\rput(2.92,1.02){$v_{11}$}
\rput(2.25,2.17){$v_3$}
\rput(2.25,-0.13){$v_{14}$}
\rput(0.92,2.17){$v_{13}$}
\rput(3.58,2.17){$v_{15}$}
\rput(2.25,-0.9){$v_{12}$}
\rput(0.25,2.56){$v_7$}
\rput(4.25,2.56){$v_8$}

\end{pspicture}\]
\caption{The first step of the colouring algorithm.}
\label{fignostub2}
\end{subfigure}

\begin{subfigure}{0.48\linewidth}
\[\begin{pspicture}(0,-0.95)(5.5,3.96)
\psline[linecolor=gray, linestyle=dashed](0.92,0.63)(2.25,-0.13)
\psline[linecolor=blue](3.58,0.63)(3.58,2.17)
\colorprint{\psline[linecolor=blue, linewidth=3pt](3.58,0.63)(3.58,2.17)
\psline[linecolor=white](3.58,0.63)(3.58,2.17)}
\psline[linecolor=gray, linestyle=dashed](2.25,2.94)(0.92,2.17)
\psline[linecolor=gray, linestyle=dashed](0.25,0.25)(2.25,-0.9)
\psline[linecolor=gray, linestyle=dashed](4.25,0.25)(2.25,-0.9)
\psline[linecolor=gray, linestyle=dashed](2.25,3.71)(0.25,2.56)
\psline[linecolor=gray, linestyle=dashed](0.25,0.25)(0.25,2.56)
\psline[linecolor=red](4.25,0.25)(4.25,2.56)
\psline[linecolor=blue](2.25,3.71)(4.25,2.56)
\colorprint{\psline[linecolor=blue, linewidth=3pt](2.25,3.71)(4.25,2.56)
\psline[linecolor=white](2.25,3.71)(4.25,2.56)}
\psline[linecolor=gray, linestyle=dashed](0.25,0.25)(0.92,0.63)
\psline[linecolor=gray, linestyle=dashed](0.92,0.63)(1.58,1.02)
\psline[linecolor=red](2.25,3.71)(2.25,2.94)
\psline[linecolor=blue](2.25,2.17)(2.25,2.94)
\colorprint{\psline[linecolor=blue, linewidth=3pt](2.25,2.17)(2.25,2.94)
\psline[linecolor=white](2.25,2.17)(2.25,2.94)}
\psline[linecolor=blue](4.25,0.25)(3.58,0.63)
\colorprint{\psline[linecolor=blue, linewidth=3pt](4.25,0.25)(3.58,0.63)
\psline[linecolor=white](4.25,0.25)(3.58,0.63)}
\psline[linecolor=red](2.92,1.02)(3.58,0.63)
\psline[linecolor=gray, linestyle=dashed](1.58,1.02)(2.25,2.17)
\psline[linecolor=blue](1.58,1.02)(2.92,1.02)
\colorprint{\psline[linecolor=blue, linewidth=3pt](1.58,1.02)(2.92,1.02)
\psline[linecolor=white](1.58,1.02)(2.92,1.02)}
\psline[linecolor=red](2.25,2.17)(2.92,1.02)

\pscircle[linecolor=black, fillstyle=solid, fillcolor=white](0.25,0.25){0.25}
\pscircle[linecolor=black, fillstyle=solid, fillcolor=white](2.25,3.71){0.25}
\pscircle[linecolor=black, fillstyle=solid, fillcolor=white](4.25,0.25){0.25}
\pscircle[linecolor=black, fillstyle=solid, fillcolor=white](0.92,0.63){0.25}
\pscircle[linecolor=black, fillstyle=solid, fillcolor=white](3.58,0.63){0.25}
\pscircle[linecolor=black, fillstyle=solid, fillcolor=white](2.25,2.94){0.25}
\pscircle[linecolor=black, fillstyle=solid, fillcolor=white](1.58,1.02){0.25}
\pscircle[linecolor=black, fillstyle=solid, fillcolor=white](2.92,1.02){0.25}
\pscircle[linecolor=black, fillstyle=solid, fillcolor=white](2.25,2.17){0.25}
\pscircle[linecolor=black, fillstyle=solid, fillcolor=white](2.25,-0.13){0.25}
\pscircle[linecolor=black, fillstyle=solid, fillcolor=white](0.92,2.17){0.25}
\pscircle[linecolor=black, fillstyle=solid, fillcolor=white](3.58,2.17){0.25}
\pscircle[linecolor=black, fillstyle=solid, fillcolor=white](2.25,-0.9){0.25}
\pscircle[linecolor=black, fillstyle=solid, fillcolor=white](0.25,2.56){0.25}
\pscircle[linecolor=black, fillstyle=solid, fillcolor=white](4.25,2.56){0.25}

\rput(0.25,0.25){$v_6$}
\rput(2.25,3.71){$v_1$}
\rput(4.25,0.25){$v_9$}
\rput(0.92,0.63){$v_5$}
\rput(3.58,0.63){$v_{10}$}
\rput(2.25,2.94){$v_2$}
\rput(1.58,1.02){$v_4$}
\rput(2.92,1.02){$v_{11}$}
\rput(2.25,2.17){$v_3$}
\rput(2.25,-0.13){$v_{14}$}
\rput(0.92,2.17){$v_{13}$}
\rput(3.58,2.17){$v_{15}$}
\rput(2.25,-0.9){$v_{12}$}
\rput(0.25,2.56){$v_7$}
\rput(4.25,2.56){$v_8$}

\end{pspicture}\]
\caption{The colouring provided by the second step of the algorithm is consistent with the one provided by the previous step.}
\label{fignostub3}
\end{subfigure}
\begin{subfigure}{0.48\linewidth}
\[\begin{pspicture}(0,-0.95)(5.5,3.96)
\psline[linecolor=blue](0.92,0.63)(2.25,-0.13)
\colorprint{\psline[linecolor=blue, linewidth=3pt](0.92,0.63)(2.25,-0.13)
\psline[linecolor=white](0.92,0.63)(2.25,-0.13)}
\psline[linecolor=blue](3.58,0.63)(3.58,2.17)
\colorprint{\psline[linecolor=blue, linewidth=3pt](3.58,0.63)(3.58,2.17)
\psline[linecolor=white](3.58,0.63)(3.58,2.17)}
\psline[linecolor=red](2.25,2.94)(0.92,2.17)
\psline[linecolor=red](0.25,0.25)(2.25,-0.9)
\psline[linecolor=blue](4.25,0.25)(2.25,-0.9)
\colorprint{\psline[linecolor=blue, linewidth=3pt](4.25,0.25)(2.25,-0.9)
\psline[linecolor=white](4.25,0.25)(2.25,-0.9)}
\psline[linecolor=blue](2.25,3.71)(0.25,2.56)
\colorprint{\psline[linecolor=blue, linewidth=3pt](2.25,3.71)(0.25,2.56)
\psline[linecolor=white](2.25,3.71)(0.25,2.56)}
\psline[linecolor=red](0.25,0.25)(0.25,2.56)
\psline[linecolor=red](4.25,0.25)(4.25,2.56)
\psline[linecolor=blue](2.25,3.71)(4.25,2.56)
\colorprint{\psline[linecolor=blue, linewidth=3pt](2.25,3.71)(4.25,2.56)
\psline[linecolor=white](2.25,3.71)(4.25,2.56)}
\psline[linecolor=blue](0.25,0.25)(0.92,0.63)
\colorprint{\psline[linecolor=blue, linewidth=3pt](0.25,0.25)(0.92,0.63)
\psline[linecolor=white](0.25,0.25)(0.92,0.63)}
\psline[linecolor=red](0.92,0.63)(1.58,1.02)
\psline[linecolor=red](2.25,3.71)(2.25,2.94)
\psline[linecolor=blue](2.25,2.17)(2.25,2.94)
\colorprint{\psline[linecolor=blue, linewidth=3pt](2.25,2.17)(2.25,2.94)
\psline[linecolor=white](2.25,2.17)(2.25,2.94)}
\psline[linecolor=blue](4.25,0.25)(3.58,0.63)
\colorprint{\psline[linecolor=blue, linewidth=3pt](4.25,0.25)(3.58,0.63)
\psline[linecolor=white](4.25,0.25)(3.58,0.63)}
\psline[linecolor=red](2.92,1.02)(3.58,0.63)
\psline[linecolor=red](1.58,1.02)(2.25,2.17)
\psline[linecolor=blue](1.58,1.02)(2.92,1.02)
\colorprint{\psline[linecolor=blue, linewidth=3pt](1.58,1.02)(2.92,1.02)
\psline[linecolor=white](1.58,1.02)(2.92,1.02)}
\psline[linecolor=red](2.25,2.17)(2.92,1.02)

\pscircle[linecolor=black, fillstyle=solid, fillcolor=white](0.25,0.25){0.25}
\pscircle[linecolor=black, fillstyle=solid, fillcolor=white](2.25,3.71){0.25}
\pscircle[linecolor=black, fillstyle=solid, fillcolor=white](4.25,0.25){0.25}
\pscircle[linecolor=black, fillstyle=solid, fillcolor=white](0.92,0.63){0.25}
\pscircle[linecolor=black, fillstyle=solid, fillcolor=white](3.58,0.63){0.25}
\pscircle[linecolor=black, fillstyle=solid, fillcolor=white](2.25,2.94){0.25}
\pscircle[linecolor=black, fillstyle=solid, fillcolor=white](1.58,1.02){0.25}
\pscircle[linecolor=black, fillstyle=solid, fillcolor=white](2.92,1.02){0.25}
\pscircle[linecolor=black, fillstyle=solid, fillcolor=white](2.25,2.17){0.25}
\pscircle[linecolor=black, fillstyle=solid, fillcolor=white](2.25,-0.13){0.25}
\pscircle[linecolor=black, fillstyle=solid, fillcolor=white](0.92,2.17){0.25}
\pscircle[linecolor=black, fillstyle=solid, fillcolor=white](3.58,2.17){0.25}
\pscircle[linecolor=black, fillstyle=solid, fillcolor=white](2.25,-0.9){0.25}
\pscircle[linecolor=black, fillstyle=solid, fillcolor=white](0.25,2.56){0.25}
\pscircle[linecolor=black, fillstyle=solid, fillcolor=white](4.25,2.56){0.25}

\rput(0.25,0.25){$v_6$}
\rput(2.25,3.71){$v_1$}
\rput(4.25,0.25){$v_9$}
\rput(0.92,0.63){$v_5$}
\rput(3.58,0.63){$v_{10}$}
\rput(2.25,2.94){$v_2$}
\rput(1.58,1.02){$v_4$}
\rput(2.92,1.02){$v_{11}$}
\rput(2.25,2.17){$v_3$}
\rput(2.25,-0.13){$v_{14}$}
\rput(0.92,2.17){$v_{13}$}
\rput(3.58,2.17){$v_{15}$}
\rput(2.25,-0.9){$v_{12}$}
\rput(0.25,2.56){$v_7$}
\rput(4.25,2.56){$v_8$}

\end{pspicture}\]
\caption{The resulting colouring connects the graph.}
\label{fignostub4}
\end{subfigure}
\caption{An example of how to construct a connecting 2-edge-colouring of a graph with no stubborn edge.  Here, $C_1=(v_1,v_2,v_3,v_4,v_5,v_6,v_7,v_1)$, $C_2=(v_1,v_8,v_9,v_{10},v_{11},v_3,v_2,v_1)$ and thus, $Q=(v_1,v_2,v_3)$.} 
\end{figure}
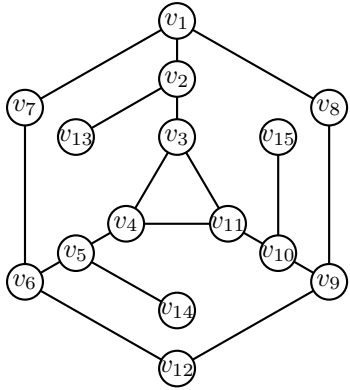
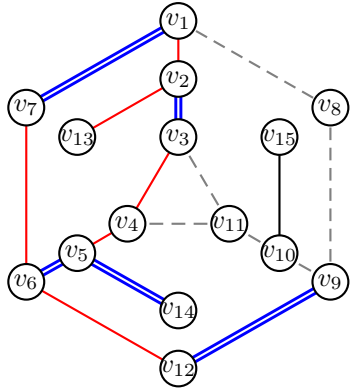
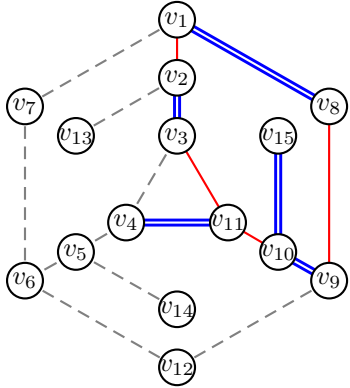
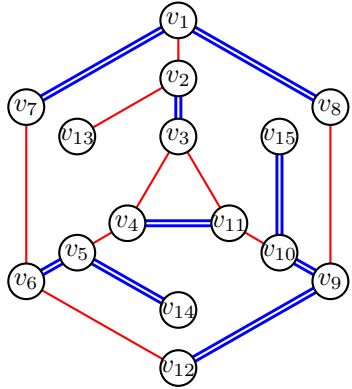
 
\vspace{0.35cm}
We want to extend the construction used in the proof of Theorem \ref{2disjcycles} to a case where the graph has no stubborn edge but where the cycles $C_1$ and $C_2$ intersect in a path $Q$. Unlike in the previous proof, there might therefore be vertices that one cannot reach from $u_i$ without using edges from $Q$, that belong to both $C_1$ and $C_2$. For example, in Figure \ref{fignostub}, one needs to use edges from $Q$ to go from $v_4$ to $v_2$ or $v_{13}$.
 
 Just like in the proof of Theorem \ref{2disjcycles}, we create the graph $H$ from $G$ by removing the edges of $C_2$ and $P_i$ but we do not remove the edges of $Q$. We then use colour 1 on all the edges that are at even distance of $u_i$ in $H$ and colour 2 and all the edges at odd distance. We do not colour the edges that cannot be reached from $u_i$ in $H$. This first step is illustrated in Figure \ref{fignostub2} (where red is colour 1, blue is colour 2 and $u_i=v_4$). Just like in the proof of Theorem \ref{2disjcycles}, this construction creates an odd properly coloured closed walk $W_1$ in $H$ with pivot $u_i$.
 
 We then create the graph $H'$ from $G$ by removing all the edges that have already be coloured but here again, we do not remove the edges of $Q$. We use colour 2 on all the edges at even distance from $u_i$ in $H'$ and colour 1 on all the edges at odd distance from $u_i$. This second step is illustrated in Figure \ref{fignostub3}. Since $H'$ contains the odd cycle $C_2$, our proof that this steps creates an odd properly coloured closed walk $W_2$ in $H'$ with pivot $u_i$ still holds.
 
 The same reasoning as before applies to prove that such a colouring would connect the graph. However, this construction colours twice the edges of $Q$ and is therefore only possible if those two colourings are compatible. Hence, we must prove for every edge $e\in Q$ that the distance from $u_i$ to $e$ in $H$ has different parity than the distance from $u_i$ to $e$ in $H'$ in order to ensure that the two steps give the same colour to $e$.
  
We denote by $d_{G'}$ the distance in a subgraph $G'$ of $G$. We denote by $a\equiv b$ the fact that $a$ and $b$ have same parity. This comes down to saying that $a+b\equiv 0$. We want to prove that $d_{H}(u_i,e)$ and $d_{H'}(u_i,e)$ have different parity. In other words, we want to prove that $d_{H}(u_i,e)+ d_{H'}(u_i,e)\equiv 1$.

For every vertex $u$ and edge $e$ of a cycle $C$, $C$ defines two paths leading from $u$ to $e$ that do not use the edge $e$. Note that $C$ is the concatenation of those two paths and the edge $e$ of length 1. Hence, if $C$ is odd, we know that the length of those two paths have same parity and thus, the same parity as $d_C(u,e)$. This leads us to a few useful observations:

\begin{itemize}
 \item For every pair of vertices $u$ and $v$ and edge $e$ of an odd cycle $C$, $d_C(u,e)\equiv d_{C-e}(u,v)+d_C(v,e)$. Indeed, one of the paths that $C$ defines between $u$ and $e$ goes through $v$ and we know that its length has the parity of $d_C(u,e)$. This path is the concatenation of the shortest (and unique) path between $u$ and $v$ in $C-e$ and a path that has the parity of $d_C(v,e)$. The claim follows.
 \item For every edge $e$ of $Q$, $d_H(u_i,e)\equiv d_{C_1}(u_i,e)$: let $P$ be the shortest path in $H$ between $u_i$ and $e$ and let $v_e$ be the endpoint of $P$ in $e$. We know that $C_1$ defines a walk $W$ from $u_i$ to $v_e$ that does not use $e$ and whose length has the parity of $d_{C_1}(u_i,e)$. If $P$ (and thus, $d_H(u_i,e)$) does not have the same parity as $W$, then $P$ and $W$ form an odd cycle in $H$ that does not use $e$, and whose intersection with $C_2$ is a proper subpath of $P$, contradicting the hypothesis that $C_1$ and $C_2$ are odd cycles that intersect in a minimal path.
 \item We can prove similarly that $d_{H'(u_i,e)}\equiv d_{C_2\cup P_i}(u_i,e)$ by considering the walks that the odd closed walk $P_i+C_2+P_i$ defines between $u_i$ and the endpoint of $e$.
\end{itemize}

Let $e$ be an edge of $Q$. Our goal is now to prove that \\ $d_{C_1}(u_i,e)+ d_{C_2\cup P_i}(u_i,e)\equiv 1$.

Let $a$ be an end vertex of $Q$. We also know that $a$ belongs to one of $\mathscr C$ and $\mathscr C'$. By symmetry, we may assume that $a\in \mathscr C$. Let $l_{\mathscr C}$ and $l_{P_i}$ be the length of $\mathscr C$ and $P_i$ respectively.

Note that $d_{C_1}(u_i,e)+ d_{C_2\cup P_i}(u_i,e)\equiv d_{C_1}(u_i,e)+ l_{P_i} +d_{C_2}(v_i,e)$ 

$\equiv d_{C_1-e}(u_i,a)+d_{C_1}(a,e)+ l_{P_i} +d_{C_2-e}(v_i,a)+d_{C_2}(a,e)$.

Since the two paths between a given vertex and a given edge in an odd cycle have same parity, $d_{C_1}(a,e)\equiv d_Q(a,e)\equiv d_{C_2}(a,e)$ and thus, $d_{C_1}(a,e)+d_{C_2}(a,e)\equiv 0$. This leaves us with $d_{C_1}(u_i,e)+ d_{C_2+P_i}(u_i,e)\equiv d_{C_1-e}(u_i,a)+ l_{P_i} +d_{C_2-e}(v_i,a)$.

Finally, observe that the concatenation of $P_i$, the walk from $u_i$ to $a$ in $C_1-e$ and the walk from $a$ to $v_i$ in $C_2-e$ is exactly the odd cycle $\mathscr C$. Thus $d_{H}(u_i,e)+ d_{H'}(u_i,e)\equiv l_{\mathscr C}\equiv 1$, which concludes the proof.\qedhere

\end{proof}

Note that the proofs of Theorem \ref{2disjcycles} and \ref{thmnostubborn} are constructive and if a connected graphs has no stubborn edge, a connecting 2-edge-colouring can be constructed in polynomial time.

\subsection{Characterization of the graphs that can be connected with two colours}

\begin{theorem}\label{thmnonbipartite}
 A connected non-bipartite graph $G$ can be connected with two colours if and only if there exists an $\mathscr S$-free component $\mathcal K$ of $G$ such that $G\setminus \mathcal K$ is empty or can be made 2-edge-connected by adding at most one edge.
\end{theorem}

\begin{proof}
 The fact that this condition is necessary follows quickly from Proposition \ref{propCN}. Indeed, let $G$ be a 2-edge-coloured graph. We know that there exists an $\mathscr S$-free component $\mathcal K$ of $G$ such that no two vertices $u$ and $v$ of $G\setminus\mathcal K$ can be connected by a properly coloured walk using a vertex of $\mathcal K$, which means that $G\setminus \mathcal K$ has to be properly connected. Since the stubborn edges that disconnect $\mathcal K$ from the rest of the graph have an endpoint in $\mathcal K$, they do not belong to $G\setminus \mathcal K$, which means that $G\setminus \mathcal K$ contains no odd closed walk and is therefore bipartite. Theorem \ref{thmbipartite} thus implies that the condition of Theorem \ref{thmnonbipartite} is necessary.
 
 Let us now prove that this condition is also sufficient. Let $\mathcal K$ be an $\mathscr S$-free component of $G$ such that $G\setminus \mathcal K$ is empty or can be made 2-edge-connected by adding at most one edge. If $G$ has no stubborn edge (in which case we have $\mathcal K=V(G)$ and $G\setminus\mathcal K$ is empty), Theorem \ref{thmnostubborn} implies that the graph can be connected with two colours. Otherwise, by Theorem \ref{thmflexcompo}, we know that $\mathcal K$ contains exactly two vertices $u_1$ and $u_2$ that are the endpoints of stubborn edges $e_1$ and $e_2$. We call $w_1$ and $w_2$ the other endpoint of $e_1$ and $e_2$. Note that it may happen that the graph only contains one stubborn edge, in which case $u_2=w_1$, $w_2=u_1$ and $e_1=e_2$.
 
 Since $G\setminus \mathcal K$ does not contain the stubborn edges that connect $\mathcal K$ to $G\setminus\mathcal K$, it cannot contain any odd closed walk and is therefore bipartite. Hence, by Theorem \ref{thmbipartite}, we know that $G\setminus\mathcal K$ can be connected with only two colours. We colour the edges of $G\setminus\mathcal K$ according to such a colouring and we want to prove that we can extend it to connect the entire graph $G$. We refer the reader to Figures \ref{fig1se} and \ref{fig2se} for an illustration of our construction in the cases where the graph has one or several stubborn edges respectively.
 
 If $G\setminus\mathcal K$ is non-empty, we know that there is a properly coloured walk $W$ between $w_1$ and $w_2$ and we colour $e_1$ and $e_2$ so that $e_1We_2$ is properly coloured too. If $G\setminus\mathcal K$ is empty, this means that $G$ has only one stubborn edge and we pick its colour arbitrarily. This step is illustrated in Figures \ref{fig1se2} and \ref{fig2se2}.
 
 If $G$ has several stubborn edges, then every odd closed walk consists of a walk in $\mathcal K$ between $u_1$ and $u_2$, $e_2$, a walk in $G\setminus\mathcal K$ between $w_2$ and $w_1$ and $e_1$. Since no edge of $\mathcal K$ appears in every odd closed walk (Proposition \ref{propsevse}), we know by Menger's theorem that there exists two edge-disjoint walks $W_1$ and $W_2$ in $\mathcal K$ that connect $u_1$ and $u_2$. Since they do not use stubborn edges, $W_1$ and $W_2$ cannot form an odd closed walk and therefore must have same parity. If $G$ has only one stubborn edge, an odd closed walk in $G$ consists of the stubborn edge $e_1$ and an even walk between $u_1$ and $u_2$ that avoids $e_1$. Here again, since there is no other stubborn edge, we find that there are two edge-disjoint walks $W_1$ and $W_2$ of same parity between $u_1$ and $u_2$ that use no stubborn edge.
 
We create the graph $H$ from $\mathcal K$ (or from $G-e$ if $G$ has only one stubborn edge) by removing all the edges of $W_2$. We colour the edges at even distance from $u_2$ with the opposite colour of the one we used on $e_2$ and we colour the edges at odd distance from $u_2$ with the colour of $e_2$. Since $W_1$ is included in $H$, we know that $u_1$ is reachable from $u_2$ and therefore, that this step of the algorithm creates a properly coloured walk $W'_1$ from $u_2$ to $u_1$ (every shortest path between $u_1$ and $u_2$ in $H$ is actually properly coloured). Also note that we have chosen the colour so that $e_2W'_1$ is properly coloured. This step is illustrated in Figures \ref{fig1se3} and \ref{fig2se3}.

We then create the graph $H'$ by removing all the already coloured edges of the graph. We colour the edges at even distance from $u_1$ with the opposite colour of the one we used on $e_1$ and we colour the edges at odd distance from $u_1$ with the colour of $e_1$. Since $W_2$ is included in $H'$, we know that this step of the algorithm creates another properly coloured walk $W'_2$ from $u_1$ to $u_2$. Here again, $e_1W'_2$ is properly coloured. This step is illustrated in Figures \ref{fig1se4} and \ref{fig2se4}.

If the graphs has several stubborn edges, recall that both $e_1We_2$ and $e_2W'_1$ are properly coloured. Hence, the closed walk $\mathscr C_1=e_1We_2W'_1$ is properly coloured and is odd since it contains the stubborn edge $e_1$ exactly once. 
Similarly, $\mathscr C_2=e_2We_1W'_2$ is odd and properly coloured too. The pivots of $\mathscr C_1$ and $\mathscr C_2$ are respectively $u_1$ and $u_2$. If the graph only has one stubborn edge, we prove similarly that 
$\mathscr C_1=e_1W'_1$ and $\mathscr C_2=e_1W'_2$ are properly coloured odd cycles of respective pivot $u_1$ and $u_2$.

    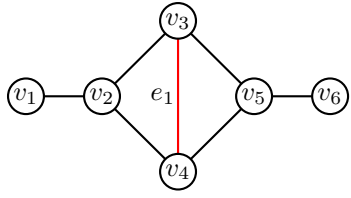
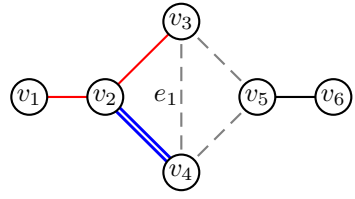
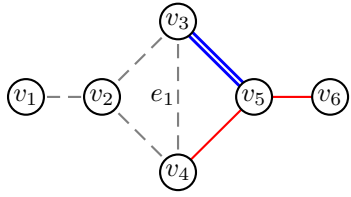
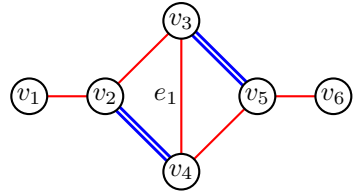
\begin{figure}[!h]
   \begin{subfigure}{0.48\linewidth}
\[\begin{pspicture}(4.5,2.5)

\psline(0.25,1.25)(1.25,1.25)
\psline(1.25,1.25)(2.25,0.25)
\psline(1.25,1.25)(2.25,2.25)
\psline[linecolor=red](2.25,0.25)(2.25,2.25)
\psline(2.25,0.25)(3.25,1.25)
\psline(2.25,2.25)(3.25,1.25)
\psline(3.25,1.25)(4.25,1.25)

 \pscircle[linecolor=black, fillstyle=solid, fillcolor=white](0.25,1.25){0.25}
 \pscircle[linecolor=black, fillstyle=solid, fillcolor=white](1.25,1.25){0.25}
 \pscircle[linecolor=black, fillstyle=solid, fillcolor=white](2.25,0.25){0.25}
 \pscircle[linecolor=black, fillstyle=solid, fillcolor=white](2.25,2.25){0.25}
 \pscircle[linecolor=black, fillstyle=solid, fillcolor=white](3.25,1.25){0.25}
 \pscircle[linecolor=black, fillstyle=solid, fillcolor=white](4.25,1.25){0.25}
 
 \rput(2.05,1.25){$e_1$}
 \rput(0.25,1.25){$v_1$}
 \rput(1.25,1.25){$v_2$}
 \rput(2.25,0.25){$v_4$}
 \rput(2.25,2.25){$v_3$}
 \rput(3.25,1.25){$v_5$}
 \rput(4.25,1.25){$v_6$}

\end{pspicture}\]
\caption{An example of graph with only \\ one stubborn edge $e_1=v_3v_4$. Its \\ removal does not disconnect the \\ graph. We pick its colour arbitrarily.\\ We set $u_1=v_3$ and $u_2=v_4$.}
\label{fig1se2}
\end{subfigure}
   \begin{subfigure}{0.48\linewidth}
\[\begin{pspicture}(4.5,2.5)
\psline[linecolor=red](0.25,1.25)(1.25,1.25)
\psline[linecolor=blue](1.25,1.25)(2.25,0.25)
\colorprint{\psline[linecolor=blue, linewidth=3pt](1.25,1.25)(2.25,0.25)
\psline[linecolor=white](1.25,1.25)(2.25,0.25)}
\psline[linecolor=red](1.25,1.25)(2.25,2.25)
\psline[linecolor=gray, linestyle=dashed](2.25,0.25)(2.25,2.25)
\psline[linecolor=gray, linestyle=dashed](2.25,0.25)(3.25,1.25)
\psline[linecolor=gray, linestyle=dashed](2.25,2.25)(3.25,1.25)
\psline(3.25,1.25)(4.25,1.25)

 \pscircle[linecolor=black, fillstyle=solid, fillcolor=white](0.25,1.25){0.25}
 \pscircle[linecolor=black, fillstyle=solid, fillcolor=white](1.25,1.25){0.25}
 \pscircle[linecolor=black, fillstyle=solid, fillcolor=white](2.25,0.25){0.25}
 \pscircle[linecolor=black, fillstyle=solid, fillcolor=white](2.25,2.25){0.25}
 \pscircle[linecolor=black, fillstyle=solid, fillcolor=white](3.25,1.25){0.25}
 \pscircle[linecolor=black, fillstyle=solid, fillcolor=white](4.25,1.25){0.25}
 
 \rput(0.25,1.25){$v_1$}
 \rput(1.25,1.25){$v_2$}
 \rput(2.25,0.25){$v_4$}
 \rput(2.25,2.25){$v_3$}
 \rput(3.25,1.25){$v_5$}
 \rput(4.25,1.25){$v_6$}
 
 \rput(2.05,1.25){$e_1$}
\end{pspicture}\]
\caption{We set $W_1=(v_3v_2v_4)$ and $W_2=(v_3v_5v_4)$. We create $H$ by removing $e$ and $W_2$ from $G$. Since $e$ is red, the edges at even distance from $u_2=v_4$ must be blue.}
\label{fig1se3}
\end{subfigure}

\begin{subfigure}{0.48\linewidth}
\[\begin{pspicture}(4.5,2.5)

\psline[linecolor=gray, linestyle=dashed](0.25,1.25)(1.25,1.25)
\psline[linecolor=gray, linestyle=dashed](1.25,1.25)(2.25,0.25)
\psline[linecolor=gray, linestyle=dashed](1.25,1.25)(2.25,2.25)
\psline[linecolor=gray, linestyle=dashed](2.25,0.25)(2.25,2.25)
\psline[linecolor=red](2.25,0.25)(3.25,1.25)
\psline[linecolor=blue](2.25,2.25)(3.25,1.25)
\colorprint{\psline[linecolor=blue, linewidth=3pt](2.25,2.25)(3.25,1.25)
\psline[linecolor=white](2.25,2.25)(3.25,1.25)}
\psline[linecolor=red](3.25,1.25)(4.25,1.25)

 \pscircle[linecolor=black, fillstyle=solid, fillcolor=white](0.25,1.25){0.25}
 \pscircle[linecolor=black, fillstyle=solid, fillcolor=white](1.25,1.25){0.25}
 \pscircle[linecolor=black, fillstyle=solid, fillcolor=white](2.25,0.25){0.25}
 \pscircle[linecolor=black, fillstyle=solid, fillcolor=white](2.25,2.25){0.25}
 \pscircle[linecolor=black, fillstyle=solid, fillcolor=white](3.25,1.25){0.25}
 \pscircle[linecolor=black, fillstyle=solid, fillcolor=white](4.25,1.25){0.25}
 
 \rput(2.05,1.25){$e_1$}
 \rput(0.25,1.25){$v_1$}
 \rput(1.25,1.25){$v_2$}
 \rput(2.25,0.25){$v_4$}
 \rput(2.25,2.25){$v_3$}
 \rput(3.25,1.25){$v_5$}
 \rput(4.25,1.25){$v_6$}

\end{pspicture}\]
\caption{We obtain $H'$ by removing all the already-coloured and repeat the same process from $u_1=v_3$.}
\label{fig1se4}
\end{subfigure}
\begin{subfigure}{0.48\linewidth}
\[\begin{pspicture}(4.5,2.5)
\psline[linecolor=red](0.25,1.25)(1.25,1.25)
\psline[linecolor=blue](1.25,1.25)(2.25,0.25)
\colorprint{\psline[linecolor=blue, linewidth=3pt](1.25,1.25)(2.25,0.25)
\psline[linecolor=white](1.25,1.25)(2.25,0.25)}
\psline[linecolor=red](1.25,1.25)(2.25,2.25)
\psline[linecolor=red](2.25,0.25)(2.25,2.25)
\psline[linecolor=red](2.25,0.25)(3.25,1.25)
\psline[linecolor=blue](2.25,2.25)(3.25,1.25)
\colorprint{\psline[linecolor=blue, linewidth=3pt](2.25,2.25)(3.25,1.25)
\psline[linecolor=white](2.25,2.25)(3.25,1.25)}
\psline[linecolor=red](3.25,1.25)(4.25,1.25)

 \pscircle[linecolor=black, fillstyle=solid, fillcolor=white](0.25,1.25){0.25}
 \pscircle[linecolor=black, fillstyle=solid, fillcolor=white](1.25,1.25){0.25}
 \pscircle[linecolor=black, fillstyle=solid, fillcolor=white](2.25,0.25){0.25}
 \pscircle[linecolor=black, fillstyle=solid, fillcolor=white](2.25,2.25){0.25}
 \pscircle[linecolor=black, fillstyle=solid, fillcolor=white](3.25,1.25){0.25}
 \pscircle[linecolor=black, fillstyle=solid, fillcolor=white](4.25,1.25){0.25}
 
 \rput(2.05,1.25){$e_1$}
 \rput(0.25,1.25){$v_1$}
 \rput(1.25,1.25){$v_2$}
 \rput(2.25,0.25){$v_4$}
 \rput(2.25,2.25){$v_3$}
 \rput(3.25,1.25){$v_5$}
 \rput(4.25,1.25){$v_6$}

\end{pspicture}\]
\caption{The resulting colouring connects the graph.}
\label{fig1se5}
\end{subfigure}

\caption{An example of how to connect a graph that has only one stubborn edge with two colours.}
\label{fig1se}
\end{figure}

   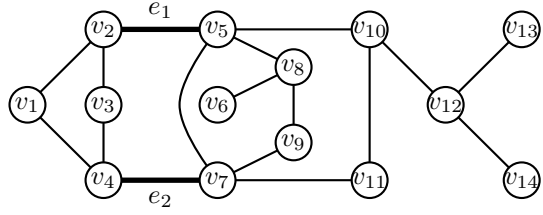
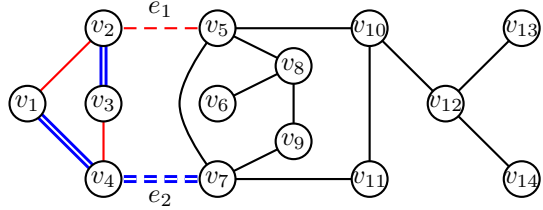
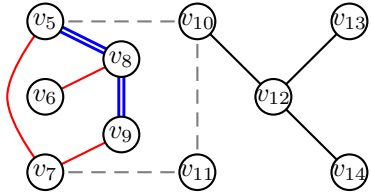
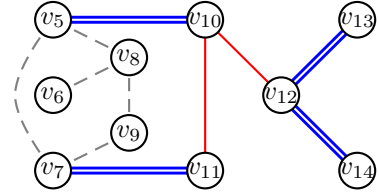
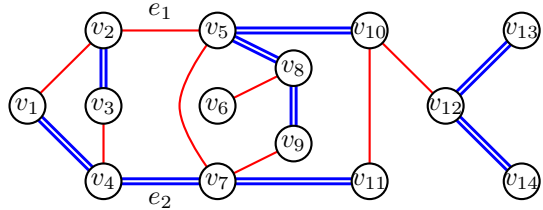
\begin{figure}[!h]
   \begin{subfigure}{\linewidth}
\[\begin{pspicture}(6.5,2.7)
\psline(0.75,0.25)(-0.25,1.25)
\psline(0.75,2.25)(-0.25,1.25)
\psline(0.75,0.25)(0.75,1.25)
\psline[linewidth=2pt](0.75,0.25)(2.25,0.25)
\psline(0.75,1.25)(0.75,2.25)
\psline[linewidth=2pt](0.75,2.25)(2.25,2.25)
\psline(2.25,0.25)(3.25,0.75)
\psline(2.25,0.25)(4.25,0.25)
\pscurve(2.25,0.25)(1.75,1.25)(2.25,2.25)
\psline(2.25,1.25)(3.25,1.75)
\psline(2.25,2.25)(3.25,1.75)
\psline(2.25,2.25)(4.25,2.25)
\psline(3.25,0.75)(3.25,1.75)
\psline(4.25,0.25)(4.25,2.25)
\psline(4.25,2.25)(5.25,1.25)
\psline(5.25,1.25)(6.25,0.25)
\psline(5.25,1.25)(6.25,2.25)

\pscircle[linecolor=black, fillstyle=solid, fillcolor=white](-0.25,1.25){0.25}
 \pscircle[linecolor=black, fillstyle=solid, fillcolor=white](0.75,0.25){0.25}
 \pscircle[linecolor=black, fillstyle=solid, fillcolor=white](0.75,1.25){0.25}
 \pscircle[linecolor=black, fillstyle=solid, fillcolor=white](0.75,2.25){0.25}
 \pscircle[linecolor=black, fillstyle=solid, fillcolor=white](2.25,0.25){0.25}
 \pscircle[linecolor=black, fillstyle=solid, fillcolor=white](2.25,1.25){0.25}
 \pscircle[linecolor=black, fillstyle=solid, fillcolor=white](2.25,2.25){0.25}
 \pscircle[linecolor=black, fillstyle=solid, fillcolor=white](3.25,0.75){0.25}
 \pscircle[linecolor=black, fillstyle=solid, fillcolor=white](3.25,1.75){0.25}
 \pscircle[linecolor=black, fillstyle=solid, fillcolor=white](4.25,0.25){0.25}
 \pscircle[linecolor=black, fillstyle=solid, fillcolor=white](4.25,2.25){0.25}
 \pscircle[linecolor=black, fillstyle=solid, fillcolor=white](5.25,1.25){0.25}
 \pscircle[linecolor=black, fillstyle=solid, fillcolor=white](6.25,0.25){0.25}
 \pscircle[linecolor=black, fillstyle=solid, fillcolor=white](6.25,2.25){0.25}

 \rput(-0.25,1.25){$v_1$}
 \rput(0.75,0.25){$v_4$}
 \rput(0.75,1.25){$v_3$}
 \rput(0.75,2.25){$v_2$}
 \rput(2.25,0.25){$v_{7}$}
 \rput(2.25,1.25){$v_6$}
 \rput(2.25,2.25){$v_5$}
 \rput(3.25,0.75){$v_{9}$}
 \rput(3.25,1.75){$v_{8}$}
 \rput(4.25,0.25){$v_{11}$}
 \rput(4.25,2.25){$v_{10}$}
 \rput(5.25,1.25){$v_{12}$}
 \rput(6.25,0.25){$v_{14}$}
 \rput(6.25,2.25){$v_{13}$}

 \rput(1.5,0){{$e_2$}}
 \rput(1.5,2.5){{$e_1$}}

\end{pspicture}\]
\caption{An example of graph with two stubborn edges $e_1=v_2v_5$ and $e_2=v_4v_7$. Their removal splits the graph into two connected components. The component $\mathcal K=\{v_5,v_6,\dots,v_{14}\}$ satisfies the condition of the theorem.}
\label{fig2se1}
\end{subfigure}
   \begin{subfigure}{\linewidth}
\[\begin{pspicture}(6.5,3)

\psline[linecolor=blue](0.75,0.25)(-0.25,1.25)
\colorprint{\psline[linecolor=blue, linewidth=3pt](0.75,0.25)(-0.25,1.25)
\psline[linecolor=white](0.75,0.25)(-0.25,1.25)}
\psline[linecolor=red](0.75,2.25)(-0.25,1.25)
\psline[linecolor=red](0.75,0.25)(0.75,1.25)
\psline[linecolor=blue, linestyle=dashed](0.75,0.25)(2.25,0.25)
\colorprint{\psline[linecolor=blue, linewidth=3pt, linestyle=dashed](0.75,0.25)(2.25,0.25)
\psline[linecolor=white, linestyle=dashed](0.75,0.25)(2.25,0.25)}
\psline[linecolor=blue](0.75,1.25)(0.75,2.25)
\colorprint{\psline[linecolor=blue, linewidth=3pt](0.75,1.25)(0.75,2.25)
\psline[linecolor=white](0.75,1.25)(0.75,2.25)}
\psline[linecolor=red, linestyle=dashed](0.75,2.25)(2.25,2.25)
\psline(2.25,0.25)(3.25,0.75)
\psline(2.25,0.25)(4.25,0.25)
\pscurve(2.25,0.25)(1.75,1.25)(2.25,2.25)
\psline(2.25,1.25)(3.25,1.75)
\psline(2.25,2.25)(3.25,1.75)
\psline(2.25,2.25)(4.25,2.25)
\psline(3.25,0.75)(3.25,1.75)
\psline(4.25,0.25)(4.25,2.25)
\psline(4.25,2.25)(5.25,1.25)
\psline(5.25,1.25)(6.25,0.25)
\psline(5.25,1.25)(6.25,2.25)

\pscircle[linecolor=black, fillstyle=solid, fillcolor=white](-0.25,1.25){0.25}
 \pscircle[linecolor=black, fillstyle=solid, fillcolor=white](0.75,0.25){0.25}
 \pscircle[linecolor=black, fillstyle=solid, fillcolor=white](0.75,1.25){0.25}
 \pscircle[linecolor=black, fillstyle=solid, fillcolor=white](0.75,2.25){0.25}
 \pscircle[linecolor=black, fillstyle=solid, fillcolor=white](2.25,0.25){0.25}
 \pscircle[linecolor=black, fillstyle=solid, fillcolor=white](2.25,1.25){0.25}
 \pscircle[linecolor=black, fillstyle=solid, fillcolor=white](2.25,2.25){0.25}
 \pscircle[linecolor=black, fillstyle=solid, fillcolor=white](3.25,0.75){0.25}
 \pscircle[linecolor=black, fillstyle=solid, fillcolor=white](3.25,1.75){0.25}
 \pscircle[linecolor=black, fillstyle=solid, fillcolor=white](4.25,0.25){0.25}
 \pscircle[linecolor=black, fillstyle=solid, fillcolor=white](4.25,2.25){0.25}
 \pscircle[linecolor=black, fillstyle=solid, fillcolor=white](5.25,1.25){0.25}
 \pscircle[linecolor=black, fillstyle=solid, fillcolor=white](6.25,0.25){0.25}
 \pscircle[linecolor=black, fillstyle=solid, fillcolor=white](6.25,2.25){0.25}

 \rput(1.5,0){{$e_2$}}
 \rput(1.5,2.5){{$e_1$}}
 \rput(-0.25,1.25){$v_1$}
 \rput(0.75,0.25){$v_4$}
 \rput(0.75,1.25){$v_3$}
 \rput(0.75,2.25){$v_2$}
 \rput(2.25,0.25){$v_{7}$}
 \rput(2.25,1.25){$v_6$}
 \rput(2.25,2.25){$v_5$}
 \rput(3.25,0.75){$v_{9}$}
 \rput(3.25,1.75){$v_{8}$}
 \rput(4.25,0.25){$v_{11}$}
 \rput(4.25,2.25){$v_{10}$}
 \rput(5.25,1.25){$v_{12}$}
 \rput(6.25,0.25){$v_{14}$}
 \rput(6.25,2.25){$v_{13}$}

\end{pspicture}\]
\caption{We use Theorem \ref{thmbipartite} to connect $G\setminus\mathcal K$ with two colours. Hence, the end vertices of the stubborn edges in $G\setminus\mathcal K$ are connected by a properly coloured walk $W=(v_2,v_3,v_4)$. We colour $e_1$ and $e_2$ so that $e_1We_2$ is properly coloured too.}
\label{fig2se2}
\end{subfigure}
   \begin{subfigure}{0.55\linewidth}
\[\begin{pspicture}(2,0)(6.5,2.7)

\psline[linecolor=red](2.25,0.25)(3.25,0.75)
\psline[linecolor=gray, linestyle=dashed](2.25,0.25)(4.25,0.25)
\pscurve[linecolor=red](2.25,0.25)(1.75,1.25)(2.25,2.25)
\psline[linecolor=red](2.25,1.25)(3.25,1.75)
\psline[linecolor=blue](2.25,2.25)(3.25,1.75)
\colorprint{\psline[linecolor=blue, linewidth=3pt](2.25,2.25)(3.25,1.75)
\psline[linecolor=white](2.25,2.25)(3.25,1.75)}
\psline[linecolor=gray, linestyle=dashed](2.25,2.25)(4.25,2.25)
\psline[linecolor=blue](3.25,0.75)(3.25,1.75)
\colorprint{\psline[linecolor=blue, linewidth=3pt](3.25,0.75)(3.25,1.75)
\psline[linecolor=white](3.25,0.75)(3.25,1.75)}
\psline[linecolor=gray, linestyle=dashed](4.25,0.25)(4.25,2.25)
\psline(4.25,2.25)(5.25,1.25)
\psline(5.25,1.25)(6.25,0.25)
\psline(5.25,1.25)(6.25,2.25)

 \pscircle[linecolor=black, fillstyle=solid, fillcolor=white](2.25,0.25){0.25}
 \pscircle[linecolor=black, fillstyle=solid, fillcolor=white](2.25,1.25){0.25}
 \pscircle[linecolor=black, fillstyle=solid, fillcolor=white](2.25,2.25){0.25}
 \pscircle[linecolor=black, fillstyle=solid, fillcolor=white](3.25,0.75){0.25}
 \pscircle[linecolor=black, fillstyle=solid, fillcolor=white](3.25,1.75){0.25}
 \pscircle[linecolor=black, fillstyle=solid, fillcolor=white](4.25,0.25){0.25}
 \pscircle[linecolor=black, fillstyle=solid, fillcolor=white](4.25,2.25){0.25}
 \pscircle[linecolor=black, fillstyle=solid, fillcolor=white](5.25,1.25){0.25}
 \pscircle[linecolor=black, fillstyle=solid, fillcolor=white](6.25,0.25){0.25}
 \pscircle[linecolor=black, fillstyle=solid, fillcolor=white](6.25,2.25){0.25}

 \rput(2.25,0.25){$v_{7}$}
 \rput(2.25,1.25){$v_6$}
 \rput(2.25,2.25){$v_5$}
 \rput(3.25,0.75){$v_{9}$}
 \rput(3.25,1.75){$v_{8}$}
 \rput(4.25,0.25){$v_{11}$}
 \rput(4.25,2.25){$v_{10}$}
 \rput(5.25,1.25){$v_{12}$}
 \rput(6.25,0.25){$v_{14}$}
 \rput(6.25,2.25){$v_{13}$}

\end{pspicture}\]
\caption{We set $W_1=(v_5,v_8,v_9,v_7)$ and $W_2\\ =(v_5,v_{10},v_{11},v_7)$. We create $H$ from $\mathcal K$ \\ by removing $W_2$. Since $e_2$ is blue, the edges \\ at even distance from $u_2=v_7$ must be \\ red. Here, $W_1$ is not properly coloured \\ but as expected,  we still create a \\ properly coloured walk $W'_1=(v_7,v_5)$ from \\ $u_2$ to $u_1$.}
\label{fig2se3}
\end{subfigure}
   \begin{subfigure}{0.42\linewidth}
\[\begin{pspicture}(2,0)(6.5,2.7)

\psline[linecolor=gray, linestyle=dashed](2.25,0.25)(3.25,0.75)
\psline[linecolor=blue](2.25,0.25)(4.25,0.25)
\colorprint{\psline[linecolor=blue, linewidth=3pt](2.25,0.25)(4.25,0.25)
\psline[linecolor=white](2.25,0.25)(4.25,0.25)}
\pscurve[linecolor=gray, linestyle=dashed](2.25,0.25)(1.75,1.25)(2.25,2.25)
\psline[linecolor=gray, linestyle=dashed](2.25,1.25)(3.25,1.75)
\psline[linecolor=gray, linestyle=dashed](2.25,2.25)(3.25,1.75)
\psline[linecolor=blue](2.25,2.25)(4.25,2.25)
\colorprint{\psline[linecolor=blue, linewidth=3pt](2.25,2.25)(4.25,2.25)
\psline[linecolor=white](2.25,2.25)(4.25,2.25)}
\psline[linecolor=gray, linestyle=dashed](3.25,0.75)(3.25,1.75)
\psline[linecolor=red](4.25,0.25)(4.25,2.25)
\psline[linecolor=red](4.25,2.25)(5.25,1.25)
\psline[linecolor=blue](5.25,1.25)(6.25,0.25)
\colorprint{\psline[linecolor=blue, linewidth=3pt](5.25,1.25)(6.25,0.25)
\psline[linecolor=white](5.25,1.25)(6.25,0.25)}
\psline[linecolor=blue](5.25,1.25)(6.25,2.25)
\colorprint{\psline[linecolor=blue, linewidth=3pt](5.25,1.25)(6.25,2.25)
\psline[linecolor=white](5.25,1.25)(6.25,2.25)}

 \pscircle[linecolor=black, fillstyle=solid, fillcolor=white](2.25,0.25){0.25}
 \pscircle[linecolor=black, fillstyle=solid, fillcolor=white](2.25,1.25){0.25}
 \pscircle[linecolor=black, fillstyle=solid, fillcolor=white](2.25,2.25){0.25}
 \pscircle[linecolor=black, fillstyle=solid, fillcolor=white](3.25,0.75){0.25}
 \pscircle[linecolor=black, fillstyle=solid, fillcolor=white](3.25,1.75){0.25}
 \pscircle[linecolor=black, fillstyle=solid, fillcolor=white](4.25,0.25){0.25}
 \pscircle[linecolor=black, fillstyle=solid, fillcolor=white](4.25,2.25){0.25}
 \pscircle[linecolor=black, fillstyle=solid, fillcolor=white](5.25,1.25){0.25}
 \pscircle[linecolor=black, fillstyle=solid, fillcolor=white](6.25,0.25){0.25}
 \pscircle[linecolor=black, fillstyle=solid, fillcolor=white](6.25,2.25){0.25}

 \rput(2.25,0.25){$v_{7}$}
 \rput(2.25,1.25){$v_6$}
 \rput(2.25,2.25){$v_5$}
 \rput(3.25,0.75){$v_{9}$}
 \rput(3.25,1.75){$v_{8}$}
 \rput(4.25,0.25){$v_{11}$}
 \rput(4.25,2.25){$v_{10}$}
 \rput(5.25,1.25){$v_{12}$}
 \rput(6.25,0.25){$v_{14}$}
 \rput(6.25,2.25){$v_{13}$}

\end{pspicture}\]
\caption{We create $H'$ by removing all the already-coloured edges. Since $e_1$ is red, the edges at even distance from $u_1=v_5$ in $H'$ must be blue. The properly coloured walk $W'_2=W_2$ connects $u_1=v_5$ to $u_2=v_7$.}
\label{fig2se4}
\end{subfigure}
   \begin{subfigure}{\linewidth}
\[\begin{pspicture}(6.5,2.7)
\psline[linecolor=blue](0.75,0.25)(-0.25,1.25)
\colorprint{\psline[linecolor=blue, linewidth=3pt](0.75,0.25)(-0.25,1.25)
\psline[linecolor=white](0.75,0.25)(-0.25,1.25)}
\psline[linecolor=red](0.75,2.25)(-0.25,1.25)
\psline[linecolor=red](0.75,0.25)(0.75,1.25)
\psline[linecolor=blue](0.75,0.25)(2.25,0.25)
\colorprint{\psline[linecolor=blue, linewidth=3pt](0.75,0.25)(2.25,0.25)
\psline[linecolor=white](0.75,0.25)(2.25,0.25)}
\psline[linecolor=blue](0.75,1.25)(0.75,2.25)
\colorprint{\psline[linecolor=blue, linewidth=3pt](0.75,1.25)(0.75,2.25)
\psline[linecolor=white](0.75,1.25)(0.75,2.25)}
\psline[linecolor=red](0.75,2.25)(2.25,2.25)
\psline[linecolor=red](2.25,0.25)(3.25,0.75)
\psline[linecolor=blue](2.25,0.25)(4.25,0.25)
\colorprint{\psline[linecolor=blue, linewidth=3pt](2.25,0.25)(4.25,0.25)
\psline[linecolor=white](2.25,0.25)(4.25,0.25)}
\pscurve[linecolor=red](2.25,0.25)(1.75,1.25)(2.25,2.25)
\psline[linecolor=red](2.25,1.25)(3.25,1.75)
\psline[linecolor=blue](2.25,2.25)(3.25,1.75)
\colorprint{\psline[linecolor=blue, linewidth=3pt](2.25,2.25)(3.25,1.75)
\psline[linecolor=white](2.25,2.25)(3.25,1.75)}
\psline[linecolor=blue](2.25,2.25)(4.25,2.25)
\colorprint{\psline[linecolor=blue, linewidth=3pt](2.25,2.25)(4.25,2.25)
\psline[linecolor=white](2.25,2.25)(4.25,2.25)}
\psline[linecolor=blue](3.25,0.75)(3.25,1.75)
\colorprint{\psline[linecolor=blue, linewidth=3pt](3.25,0.75)(3.25,1.75)
\psline[linecolor=white](3.25,0.75)(3.25,1.75)}
\psline[linecolor=red](4.25,0.25)(4.25,2.25)
\psline[linecolor=red](4.25,2.25)(5.25,1.25)
\psline[linecolor=blue](5.25,1.25)(6.25,0.25)
\colorprint{\psline[linecolor=blue, linewidth=3pt](5.25,1.25)(6.25,0.25)
\psline[linecolor=white](5.25,1.25)(6.25,0.25)}
\psline[linecolor=blue](5.25,1.25)(6.25,2.25)
\colorprint{\psline[linecolor=blue, linewidth=3pt](5.25,1.25)(6.25,2.25)
\psline[linecolor=white](5.25,1.25)(6.25,2.25)}

\pscircle[linecolor=black, fillstyle=solid, fillcolor=white](-0.25,1.25){0.25}
 \pscircle[linecolor=black, fillstyle=solid, fillcolor=white](0.75,0.25){0.25}
 \pscircle[linecolor=black, fillstyle=solid, fillcolor=white](0.75,1.25){0.25}
 \pscircle[linecolor=black, fillstyle=solid, fillcolor=white](0.75,2.25){0.25}
 \pscircle[linecolor=black, fillstyle=solid, fillcolor=white](2.25,0.25){0.25}
 \pscircle[linecolor=black, fillstyle=solid, fillcolor=white](2.25,1.25){0.25}
 \pscircle[linecolor=black, fillstyle=solid, fillcolor=white](2.25,2.25){0.25}
 \pscircle[linecolor=black, fillstyle=solid, fillcolor=white](3.25,0.75){0.25}
 \pscircle[linecolor=black, fillstyle=solid, fillcolor=white](3.25,1.75){0.25}
 \pscircle[linecolor=black, fillstyle=solid, fillcolor=white](4.25,0.25){0.25}
 \pscircle[linecolor=black, fillstyle=solid, fillcolor=white](4.25,2.25){0.25}
 \pscircle[linecolor=black, fillstyle=solid, fillcolor=white](5.25,1.25){0.25}
 \pscircle[linecolor=black, fillstyle=solid, fillcolor=white](6.25,0.25){0.25}
 \pscircle[linecolor=black, fillstyle=solid, fillcolor=white](6.25,2.25){0.25}

 \rput(1.5,0){{$e_2$}}
 \rput(1.5,2.5){{$e_1$}}
 \rput(-0.25,1.25){$v_1$}
 \rput(0.75,0.25){$v_4$}
 \rput(0.75,1.25){$v_3$}
 \rput(0.75,2.25){$v_2$}
 \rput(2.25,0.25){$v_{7}$}
 \rput(2.25,1.25){$v_6$}
 \rput(2.25,2.25){$v_5$}
 \rput(3.25,0.75){$v_{9}$}
 \rput(3.25,1.75){$v_{8}$}
 \rput(4.25,0.25){$v_{11}$}
 \rput(4.25,2.25){$v_{10}$}
 \rput(5.25,1.25){$v_{12}$}
 \rput(6.25,0.25){$v_{14}$}
 \rput(6.25,2.25){$v_{13}$}

\end{pspicture}\]
\caption{The resulting colouring connects the graph.}
\label{fig2se5}
\end{subfigure}

\caption{An example where the graph has several stubborn edges.}
\label{fig2se}
\end{figure}

 We claim that the graph is now properly connected. Indeed, let $V_1$ be the set of vertices that can be reached from $u_2$ in $H$ and let $V_2$ be its complement in $\mathcal K$. Hence, the vertices of $V_2$ can be reached from $u_2$ in $H'$ and thus from $u_1$ since $u_1$ and $u_2$ are connected in $H'$.
  \begin{itemize}
  \item Let $x$ and $y$ be two vertices of $V_1$. One can go from $x$ to $u_2$ using a shortest path in $H$, use $\mathscr C_2$ to go from $u_2$ to $u_2$ and then go from $u_2$ to $y$ using a shortest path in $H$ again. For example, in Figure \ref{fig2se5}, one can go from $v_5$ to $v_9$ by using $(v_5,v_7)$ to go to $v_7=u_2$, use $\mathscr C_2=(v_7,v_4,v_3,v_2,v_5,v_{10},v_{11},v_7)$ and then go to $v_9$ by using $(v_7,v_9)$.
  \item Let $x$ and $y$ be two vertices of $V_2$. Similarly, one can go from $x$ to $u_1$ using a shortest path in $H'$, use $\mathscr C_1$ from $u_1$ to $u_1$ and then go from $u_1$ to $y$ by a shortest path in $H'$.
  \item Let $x\in V_1$ and $y\in V_2$. One can go from $x$ to $u_2$ using a shortest path in $H$, from $u_2$ to $u_1$ using $e_1We_2$ (or just $e_1$ if the graph has one stubborn edge) and from $u_1$ to $y$ with a shortest path in $H'$. For example, in Figure \ref{fig2se5}, $v_6$ and $v_{10}$ are connected by the walk $(v_6,v_8,v_9,v_7,v_4,v_3,v_2,v_5,v_{10})$.
  \item We initialized our edge-colouring so that $G\setminus\mathcal K$ is properly connected.
  \item Let $x\in V_1$ and $y\in G\setminus\mathcal K$. Since $G\setminus\mathcal K$ is properly connected, there exists a properly coloured walk $W_3$ from $w_2$ to $y$. One can go from $x$ to $u_2$ using edges of the first search, and then, go to $w_2$ using edges of $\mathscr C_2$. Note that the two edges around $u_2$ in $\mathscr C_2$ have the same colour and are compatible with the walk we use from $x$ to $u_2$. One can thus use $\mathscr C_2$ is any direction between $u_2$ and $w_2$. Since $w_2$ is not the pivot of $\mathscr C_2$, the two edges adjacent to $w_2$ in $\mathscr C_2$ do not have the same colour and it is thus possible to choose the colour of the last edge of the walk we use between $u_2$ and $w_2$. We thus choose the walk between $u_2$ and $w_2$ so that it is then possible to use $W_3$ between $w_2$ and $y$. 
  
  For example, in Figure \ref{fig2se5}, let us try to connect $v_9\in V_1$ and $v_1\in G\setminus\mathcal K$. We go from $v_9$ to $u_2=v_7$ by shortest path and we can then use $\mathscr C_2$ in any direction. Since our colouring connects $G\setminus\mathcal K$, we know that there exists a properly coloured walk from $w_2=v_4$ to $v_1$, for example $W_3=(v_4,v_1)$. Here, $W_3$ starts with a blue edge and we therefore want to use $\mathscr C_2$ so that we arrive on $v_4$ with a red edge, which is possible since we can use $\mathscr C_2$ in any direction and the two edges incident to $w_2$ have different colours. Thus, $v_9$ and $v_1$ are connected by $(v_9,v_7,v_{11},v_{10},v_5,v_2,v_3,v_4,v_5)$. If $W_3$ started with a red edge, we could have used the other part of $\mathscr C_2$, $(v_7v_4)$, to connect them.
  \item Similarly, if $x\in V_2$ and $y\in G\setminus\mathcal K$, one can go from $x$ to $u_1$ with a shortest path in $H'$ and go from $u_1$ to $w_1$ using a subwalk of $\mathscr C_1$ that makes it possible to go from $w_1$ to $y$.\qedhere
 \end{itemize}

\end{proof}
 
 Note that this proof is constructive and provides a connecting 2-edge-colouring in polynomial time for any graph that can be connected with two colours.
 
 \section{Conclusion}
 
 Putting all together, we obtain the following theorem:
 
 \begin{theorem}\label{mainthm}
  The minimum number of colours required by a connecting edge-colouring of a graph $G$ is:
    \begin{itemize}
     \item 1 if $G$ is complete;
     \item its maximum degree $\Delta(G)$ if $G$ is a tree;
     \item 2 if $G$ is bipartite and can be made 2-edge-connected by adding at most one edge;
     \item 2 if $G$ is non-bipartite and contains an $\mathscr S$-free component $\mathcal K$ such that $G\setminus \mathcal K$ is empty or can be made 2-edge-connected by adding at most one edge;
     \item 3 otherwise
    \end{itemize}
    
    Furthermore, in every case, an optimal connecting colouring can be found in polynomial time.
 \end{theorem}
 
 Polynomial algorithms for optimal connecting colouring follows from the constructive proofs of Theorems \ref{lemmacycle}, \ref{thmbipartite} or \ref{thmnonbipartite} depending on which case occurs.

 \vspace{0.35cm}
Interesting questions for future works could be to study alternative definitions of connectivity. Indeed, the definition of the connectivity of a graph is well-agreed-upon, but there are many ways to generalize it that are no longer equivalent in walk-restricted graphs. For example, in edge-coloured graphs, the fact that there exists a properly coloured walk from any vertex $u$ to any vertex $v$ does not imply the existence a properly coloured closed walk that would start in $u$, go to $v$ and then back to $u$. This leads to the definition of colour-connectivity introduced by Saad in \cite{saad}. Only in colour-connected graphs can a single closed walk visit all the vertices of the graph. This new definition of connectivity can increase significantly the number of colours required for a connecting edge-colouring. The most extreme case is the case of graphs with vertices of degree 1. Such graphs cannot be made colour-connected, no matter how many colours are available.

Another idea of possible continuation would be to study definitions of connectivity that require the vertices to be connected by paths or trails instead of walks. Numerous papers have already studied the definition based on paths and the definition based on trails has been studied in \cite{trail}, but the complexity of the proper connection number is still open in both cases. Again, these definitions are equivalent in standard graphs but not in walk-restricted graphs (a trail is a walk that may repeat vertices but does not repeat edges). For example, we only need two colours to connect the vertices of the graph depicted in Figure \ref{concl1} with walks or trails, but in order to connect them with paths, we must give a different colour to each edge of the form $vu_i$. Similarly, two colours are enough to connect the vertices of the graph of Figure \ref{concl2} with walks but $k$ we need $k$ to connect them with trails or paths. 

 \begin{figure}[!h]
\begin{subfigure}{0.48\linewidth}
\[\begin{pspicture}(2.5,3.72)
\psline[linecolor=red](0.25,0.25)(1.25,1.25)
\psline[linecolor=red](2.25,0.25)(1.25,1.25)
\psline[linecolor=red, linestyle=dashed](1.25,0.04)(1.25,1.25)
\psline[linecolor=red, linestyle=dashed](1.05,0.04)(1.25,1.25)
\psline[linecolor=red, linestyle=dashed](1.45,0.04)(1.25,1.25)
\psline[linecolor=red](0.71,-0.06)(1.25,1.25)
\psline[linecolor=red](1.79,-0.06)(1.25,1.25)
\psline[linecolor=blue](1.25,1.25)(0.54,2.47)
\colorprint{\psline[linecolor=blue, linewidth=3pt](1.25,1.25)(0.54,2.47)
\psline[linecolor=white](1.25,1.25)(0.54,2.47)}
\psline[linecolor=blue](1.25,1.25)(1.96,2.47)
\colorprint{\psline[linecolor=blue, linewidth=3pt](1.25,1.25)(1.96,2.47)
\psline[linecolor=white](1.25,1.25)(1.96,2.47)}
\psline[linecolor=red](0.54,2.47)(1.96,2.47)

 \pscircle[fillstyle=solid, fillcolor=white](0.25,0.25){0.25}
 \pscircle[fillstyle=solid, fillcolor=white](2.25,0.25){0.25}
 \pscircle[fillstyle=solid, fillcolor=white](1.25,1.25){0.25}
 \pscircle[fillstyle=solid, fillcolor=white](0.71,-0.06){0.25}
 \pscircle[fillstyle=solid, fillcolor=white](1.79,-0.06){0.25}
\pscircle[fillstyle=solid, fillcolor=white](0.54,2.47){0.25}
\pscircle[fillstyle=solid, fillcolor=white](1.96,2.47){0.25}
\rput(1.25,-0.06){...}
\rput(0.25,0.25){$u_1$}
 \rput(2.25,0.25){$u_k$}
 \rput(1.25,1.25){$v$}
 \rput(0.71,-0.06){$u_2$}
 \rput(1.79,-0.06){...}

\end{pspicture}\]
\caption{An example of graphs that only \\ requires 2 colours to be connected \\ by walks or trails but $k$ colours to be \\ connected by paths.}
\label{concl1}
\end{subfigure}
\begin{subfigure}{0.48\linewidth}
\[\begin{pspicture}(2.5,3.72)
\psline[linecolor=blue](1.25,1.25)(1.25,2.25)
\colorprint{\psline[linecolor=blue, linewidth=3pt](1.25,1.25)(1.25,2.25)
\psline[linecolor=white](1.25,1.25)(1.25,2.25)}
\psline[linecolor=red](0.25,0.25)(1.25,1.25)
\psline[linecolor=red](2.25,0.25)(1.25,1.25)
\psline[linecolor=red, linestyle=dashed](1.25,0.04)(1.25,1.25)
\psline[linecolor=red, linestyle=dashed](1.05,0.04)(1.25,1.25)
\psline[linecolor=red, linestyle=dashed](1.45,0.04)(1.25,1.25)
\psline[linecolor=red](0.71,-0.06)(1.25,1.25)
\psline[linecolor=red](1.79,-0.06)(1.25,1.25)
\psline[linecolor=red](1.25,2.25)(0.54,3.47)
\psline[linecolor=red](1.25,2.25)(1.96,3.47)
\psline[linecolor=blue](0.54,3.47)(1.96,3.47)
\colorprint{\psline[linecolor=blue, linewidth=3pt](0.54,3.47)(1.96,3.47)
\psline[linecolor=white](0.54,3.47)(1.96,3.47)}

 \pscircle[fillstyle=solid, fillcolor=white](0.25,0.25){0.25}
 \pscircle[fillstyle=solid, fillcolor=white](2.25,0.25){0.25}
 \pscircle[fillstyle=solid, fillcolor=white](1.25,1.25){0.25}
 \pscircle[fillstyle=solid, fillcolor=white](1.25,2.25){0.25}
 \pscircle[fillstyle=solid, fillcolor=white](0.71,-0.06){0.25}
 \pscircle[fillstyle=solid, fillcolor=white](1.79,-0.06){0.25}
\pscircle[fillstyle=solid, fillcolor=white](0.54,3.47){0.25}
\pscircle[fillstyle=solid, fillcolor=white](1.96,3.47){0.25}
\rput(1.25,-0.06){...}
\rput(0.25,0.25){$u_1$}
 \rput(2.25,0.25){$u_k$}
 \rput(1.25,1.25){$v$}
 \rput(0.71,-0.06){$u_2$}
 \rput(1.79,-0.06){...}

\end{pspicture}\]
\caption{An example of graphs that only requires 2 colours to be connected by walks but $k$ colours to be connected by trails or paths.}
\label{concl2}
\end{subfigure}
\caption{}
\end{figure}

Another interesting problem would be to study the complexity of extending a partial edge-colouring of a graph: given a partial edge-colouring using at most $k$ colours, is it possible to extend it into a connecting $k$-edge-colouring of the graph?

Finally, it could also be interesting to study the stretch of our connecting edge-colouring. The \textbf{stretch} is the maximum ratio between the length of the shortest walk between two vertices in the original unrestricted graph and in the restricted graphs. For example, the colouring depicted in Figure \ref{mainexample} connects the graph but the vertices $v_0$ and $v_2$ are at distance 9 in the edge-coloured graph while their distance is only 2 in the uncoloured graph, which means that the stretch of this edge-colouring is at least $\frac 9 2$. Interesting questions could therefore be to determine the number of colours required for a connecting edge-colouring of stretch bounded by a given $k$, or to find a connecting colouring of minimum stretch with a given number of colours. Previous papers have already studied the problem of \textbf{strong proper connection number} where every pair of vertices has to be connected by properly-coloured shortest paths \cite{huangyuan} \cite{lumduanhom}. This comes down to finding the smallest number of colours such that there exists a connecting colouring of stretch 1.

Of course, all the above questions also make sense in directed graphs.

\bibliographystyle{alpha}

\newcommand{\etalchar}[1]{$^{#1}$}

\end{document}